\documentclass[11pt]{article}
\usepackage[letterpaper,margin=0.92in]{geometry}
\usepackage[T1]{fontenc}
\usepackage{graphicx,needspace,placeins,float}
\usepackage{amsmath,amssymb,amsthm,mathtools,booktabs,tabularx,microtype,hyperref,enumitem}
\hypersetup{colorlinks=true,linkcolor=blue,urlcolor=blue,citecolor=blue,
  pdftitle={Superpolynomial Lists for Reed--Solomon Codes over Binary Extension Fields}}
\newcommand{\F}{\mathbb F}
\newcommand{\Tr}{\operatorname{Tr}}
\newcommand{\wt}{\operatorname{wt}}
\newcommand{\RS}{\operatorname{RS}}
\newcommand{\gbinom}[3]{\genfrac{[}{]}{0pt}{}{#1}{#2}_{#3}}
\newtheorem{theorem}{Theorem}[section]
\newtheorem{lemma}[theorem]{Lemma}
\newtheorem{proposition}[theorem]{Proposition}
\newtheorem{corollary}[theorem]{Corollary}
\theoremstyle{definition}
\newtheorem{definition}[theorem]{Definition}
\newtheorem{example}[theorem]{Example}
\theoremstyle{remark}
\newtheorem{remark}[theorem]{Remark}
\title{Superpolynomial Lists for Reed--Solomon Codes\\over Binary Extension Fields}
\author{Minshen Zhu\\
\small Independent Researcher\\
\small \texttt{minshen.zh@gmail.com}}
\date{}
\begin{document}
\maketitle
\begin{abstract}
Recent work shows that ordinary Reed--Solomon codes over prime fields admit polynomial-size lists up to list-decoding capacity on every evaluation set. We show a contrasting phenomenon over binary extension fields. For every fixed rate $0<R<1$, there exists $\varepsilon_R>0$ such that, for all sufficiently large $h$, the full-field Reed--Solomon code over $\F_{2^h}$ has a received word with $n^{\Omega_R(\log n)}$ codewords within relative distance $1-R-\varepsilon_R$. To our knowledge, these are the first superpolynomial ordinary Reed--Solomon list-size lower bounds at fixed positive rate and a fixed positive agreement gap above the rate. Thus capacity-approaching polynomial list size for Reed--Solomon codes is not uniform over the underlying field, and the obstruction occurs at every fixed rate. As a puncturing consequence, the same holds for primitive Reed--Solomon codes; the obstruction is also robust under controlled dense puncturing. The proof begins with an exact correspondence that turns certain binary Bose--Chaudhuri--Hocquenghem weight-enumerator coefficients into sizes of single Reed--Solomon lists. A classical binary weight family gives the initial quarter-rate construction, while a quadratic-form enumeration supplies its superpolynomial multiplicity. A common-kernel quotient/descent argument and common-root padding then extend the construction to all rates. We also determine the optimized agreement envelope achieved by the construction.
\end{abstract}


\section{Introduction}\label{sec:introduction}
Reed--Solomon (RS) codes are among the most basic and widely used error-correcting codes. A rate-$R$ RS code has relative minimum distance about $1-R$, so unique decoding can correct nearly a $(1-R)/2$ fraction of adversarial errors. List decoding asks for all codewords within a larger radius of a received word. Its natural endpoint is the list-decoding capacity radius $1-R$: for a fixed slack $\varepsilon>0$, the target is to bound the number of codewords within relative distance $1-R-\varepsilon$ by a polynomial in the block length. Equivalently, in agreement language, one asks for polynomially many codewords agreeing with a received word on more than an $R+\varepsilon$ fraction of coordinates.

Two canonical evaluation sets are the full field and the set of nonzero field elements; the latter gives the \emph{primitive} RS code of length $q-1$ over $\F_q$. This leads to a basic field-uniform question:
\begin{quote}
\emph{Is there any fixed positive rate $R$ for which these canonical RS families over all finite fields have list size $n^{O_{R,\varepsilon}(1)}$ at radius $1-R-\varepsilon$, for every fixed $\varepsilon>0$?}
\end{quote}
For these ordinary RS families, the classical algorithmic guarantee is due to Sudan and Guruswami--Sudan: polynomial-time list decoding beyond the unique-decoding radius and up to the Johnson radius $1-\sqrt R$, independently of the evaluation set~\cite{Sudan1997,GuruswamiSudan}. Over prime fields, recent results go much further without changing the code or its evaluation set. Brakensiek, Chen, Putterman, Zhang, and Zheng proved capacity-approaching algorithmic bounds at every fixed rate for arbitrary prescribed evaluation sets when $q\ge C_{R,\varepsilon}n$, with list size $q^{O_{R,\varepsilon}(1)}$~\cite{BCPZZ}. Jeronimo established an $n^{O_\varepsilon(1)}$ list bound independent of $q$ for every prime $q\ge n$ and every prescribed evaluation set~\cite{Jeronimo}. In particular, the canonical full-field and primitive RS codes over prime fields have polynomial-size lists at every fixed positive agreement gap above the rate.

Other capacity-oriented results address related but different code families or evaluation sets. Parvaresh--Vardy codes improved the Johnson decoding radius at sufficiently low rates~\cite{PV}, while folded RS codes subsequently achieved list-decoding capacity at every fixed rate; further improvements include folded and multiplicity codes~\cite{GR,KRSW}. Alternatively, generic evaluation sets and random puncturings yield capacity-achieving ordinary RS codes under suitable alphabet-size conditions~\cite{BGM,AGLZ}. These results do not, by themselves, establish the same guarantee for the canonical evaluation sets over extension fields.

For extension fields, the picture has been less clear. Guruswami and Rudra gave explicit superpolynomial ordinary-RS lists in a vanishing-rate regime, with dimension $k=n^\delta$~\cite{GRlimits}. Ben-Sasson, Kopparty, and Radhakrishnan substantially strengthened this phenomenon using subspace polynomials: their superpolynomial regime likewise has rate tending to zero, while a constant-rate specialization gives polynomially large bad lists at certain dyadic rates~\cite{BKR}. At fixed rate, Hamming-ball averaging forces exponential
lists with agreement $k/n+\Theta(1/\log q)$, where
$k=\lfloor Rn\rfloor$, even for some exponentially large
alphabets (Appendix~\ref{app:hamming-ball-count}).
This is a shrinking-gap rather than a fixed-gap obstruction.
Related shrinking-window results appear in
Crites--Stewart~\cite{CS} and Krachun--Kazanin--Hab\"ock~\cite{KKH}.

Thus, to our knowledge, before this work no superpolynomial
ordinary-RS list was known at both fixed positive rate and
fixed positive agreement gap. On the positive side, generic and random
evaluation-set theorems did not settle the canonical full-field
or primitive codes, while guarantees for every prescribed
evaluation set were available over prime fields. Whether
those canonical codes retain polynomial list size near
capacity over extension fields remained open.

\subsection{Our result}\label{sec:informal-result}
\begin{theorem}[Informal]\label{thm:informal}
For every fixed $0<R<1$, there exists $\varepsilon_R>0$ such that, for all sufficiently large $h$, the full-field RS code of dimension $\lfloor R2^h\rfloor$ over $\F_{2^h}$ has a received word with
\[
 2^{h^2/8-O_R(h)}
 =n^{\frac18\log_2 n-O_R(1)}
 =n^{\Omega_R(\log n)}
\]
codewords within relative distance $1-R-\varepsilon_R$, where $n=2^h$.
\end{theorem}

The formal statement, Theorem~\ref{thm:allrate}, gives
an explicit agreement threshold $\alpha_*(R)>R$.
For any fixed $0<\varepsilon_R<\alpha_*(R)-R$, it yields
Theorem~\ref{thm:informal} at agreement $R+\varepsilon_R$.

The theorem resolves this gap in the negative, at \emph{every}
fixed rate and for every sufficiently large binary extension degree.
Unlike the Hamming-ball obstruction, these lists persist at a
fixed positive agreement gap. Together with the prime-field
results~\cite{BCPZZ,Jeronimo}, this shows that polynomial list
size for ordinary RS codes cannot be characterized by rate
and asymptotic alphabet-size scaling alone: field arithmetic and the evaluation
set matter. It also redirects the questions ahead: which
field families and evaluation sets admit polynomial lists
near capacity, how robust are the counterexamples under
puncturing, and how close can their agreement approach the
Johnson bound?

\paragraph{Landscape.}
Table~\ref{tab:radius} summarizes the parts of the list-decoding landscape most relevant to the field-uniform question. Figure~\ref{fig:bounds} shows the quantitative agreement envelope of our binary-extension construction.
\begin{table}[H]\centering
\small
\renewcommand{\arraystretch}{1.16}
\begin{tabularx}{\linewidth}{@{}>{\raggedright\arraybackslash}p{1.57in}>{\raggedright\arraybackslash}p{1.30in}>{\raggedright\arraybackslash}p{1.34in}>{\raggedright\arraybackslash}X@{}}
\toprule
Result & Agreement $\alpha$ & List size & Restrictions \\
\midrule
\multicolumn{4}{@{}l}{\emph{Upper bounds}} \\
Guruswami--Sudan~\cite{GuruswamiSudan}
 & $>\sqrt R$
 & $\le n^{O_{R,\alpha}(1)}$
 & --- \\
Brakensiek et al.~\cite{BCPZZ}
 & $\ge R+\varepsilon$
 & $\le q^{O_{R,\varepsilon}(1)}$
 & Prime fields; $q\ge C_{R,\varepsilon}n$ \\
Jeronimo~\cite{Jeronimo}
 & $\ge R+\varepsilon$
 & $\le n^{O_\varepsilon(1)}$
 & Prime fields; $q\ge n$ \\
\addlinespace[3pt]
\multicolumn{4}{@{}l}{\emph{Lower bounds}} \\
Ben-Sasson et al. (BKR)~\cite{BKR}
 & $2R$
 & $\ge n^{\Omega(1)}$
 & Selected extension fields; $R=2^{-r}$, $r\ge3$ \\
$q$-ary capacity/counting obstruction
 & $k/n+\Theta(1/\log q)$
 & $\ge 2^{\Omega(n)}$
 & $k=\lfloor Rn\rfloor$; $n\le q\le2^{\gamma n}$ \\
This paper
 & $R+\varepsilon_R$
 & $\ge n^{\Omega_R(\log n)}$
 & Binary extensions; full field and primitive \\
\bottomrule
\end{tabularx}
\caption{List-size bounds for ordinary RS codes ($n$: block length;
$q$: field size). Upper bounds hold for every received word;
lower bounds exhibit received words with large lists. The counting
row assumes fixed $R$ and sufficiently small $\gamma>0$.}
\label{tab:radius}
\end{table}

\begin{figure}[tb]
\centering
\includegraphics[width=0.94\linewidth]{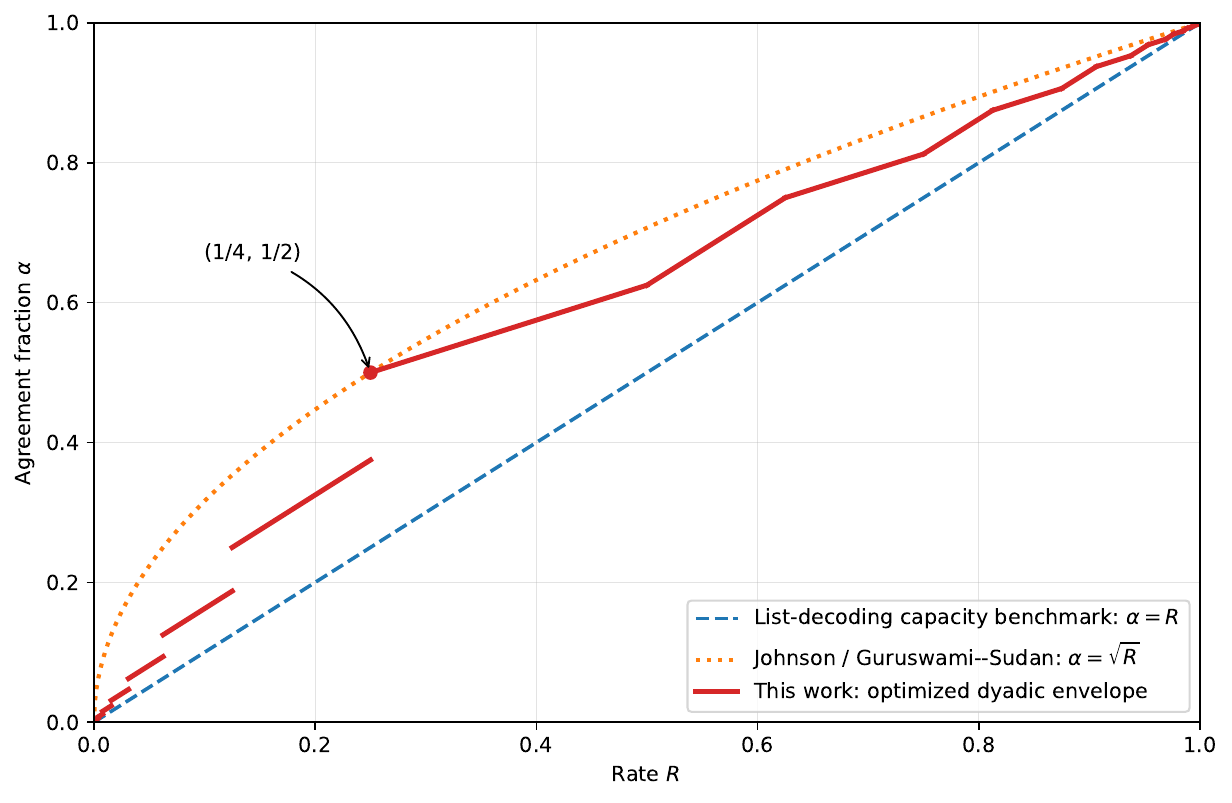}
\caption{Agreement versus rate for ordinary RS codes. The capacity benchmark is $\alpha=R$, while Guruswami--Sudan guarantees polynomial lists above $\alpha=\sqrt R$. The red curve shows the limiting agreement $\alpha=\alpha_*(R)$ achieved by our construction, defined in~\eqref{eq:optimized-agreement-envelope}. For every fixed rate, each fixed agreement strictly between $R$ and this curve has superpolynomial list size for all sufficiently large binary extension degrees. The point $(1/4,1/2)$ comes from the initial full-field construction.}
\label{fig:bounds}
\end{figure}
\FloatBarrier

\subsection{Technical overview: three stages}\label{sec:intro-mechanism}
The proof has three stages. First, a classical BCH weight-enumerator coefficient becomes a large RS list through an exact correspondence. Stage~1 uses this enumerative input only numerically. Second, common-root padding shows how far this single full-field anchor can be moved without using any further structure. Third, we use the translation symmetries of Schmidt's quadratic-form family to obtain lists that remain superpolynomial on smaller additive evaluation sets. Their smaller polynomial degrees leave more room for padding to the target full-field rate.

\paragraph{Stage 1: a BCH weight enumerator becomes an RS list.}
Let $K=\F_{2^h}$. Suppose an $\F_2$-valued polynomial $c(X)$ has even weight $2\le D<2^h$ on $K$ and degree at most $2^h-D$. Its $2^h-D$ zeros saturate the degree bound, so $c(X)$ divides $X^{2^h}-X$. A derivative calculation then forces the complementary factor to have the form
\[
  \frac{X^{2^h}-X}{c(X)}=p(X)^2+X.
\]
The $D$ positions where $c$ equals one are exactly the field elements $\beta$ satisfying $p(\beta)=\beta^{2^{h-1}}$. It follows that the degree-$D/2$ polynomial $p(X)$ has $D$ agreements with the word $\rho(X)=X^{2^{h-1}}$ over $K$. This construction is reversible, so the corresponding BCH weight-enumerator coefficient is exactly the size of one full-field RS list (Proposition~\ref{prop:bijection}).

Schmidt's quadratic-form enumeration yields $n^{\Omega(\log n)}$ such BCH words of weight $(1/2-o(1))2^h$; the classical weight-enumerator literature provides the historical background~\cite{Schmidt2020,Berlekamp}. Under the correspondence, they become a quarter-rate RS list with limiting agreement $1/2$.

\paragraph{Stage 2: padding reveals the need for lower-rate anchors.}
A natural way to raise the rate is to multiply \emph{every} message and their common received word by the same polynomial. Its roots become automatic agreements, so we have not lost the common center or identified any messages. This operation can only raise the message degree. Consequently, the quarter-rate anchor cannot reach any target rate $R<1/4$ by padding alone; that regime simply requires an anchor of lower rate.

The need for lower-rate anchors also appears, less trivially, at high rates. Each added root raises the message degree by one but creates a \emph{new} agreement only if that coordinate was previously a disagreement. Starting from a full-field seed with rate $r_0$ and agreement $\alpha_0$, reaching rate $R$ requires about $(R-r_0)2^h$ padding roots. A uniformly chosen root is a former disagreement with probability $1-\alpha_0$. Averaging over a \emph{single common} padding set retains at least a $1/2^h$ fraction of the original list at every fixed threshold strictly below
\[
  \alpha(R)=\alpha_0+(R-r_0)(1-\alpha_0),\qquad
  \alpha(R)-R=\alpha_0(1-R)-r_0(1-\alpha_0).
\]
For the quarter-rate anchor $(r_0,\alpha_0)=(1/4,1/2)$, the excess is $(3-4R)/8$. Thus padding this anchor reaches every rate $1/4\le R<3/4$, but no longer gives a positive fixed gap once $R\ge3/4$. Extending the construction to both missing rate regimes therefore calls for lower-rate anchors.

\paragraph{Stage 3: quotienting structured BCH families.}
To obtain anchors with smaller full-field degree cost, we return to the BCH family and use more than its weight count. Schmidt's trace-quadratic description yields a superpolynomial subfamily whose members are $W_f$-invariant for large additive subspaces $W_f\le K$, meaning that
\[
   f(\beta+\omega)=f(\beta)
   \qquad(\beta\in K,\ \omega\in W_f).
\]
A double count finds one fixed $c$-dimensional subspace $W$ under which superpolynomially many words are $W$-invariant. The common subspace polynomial gives a descent to an additive evaluation domain $V$ of size $\lambda n$, where $n=2^h$ and $\lambda=2^{-c}$; this is related to classical Kasami--Lin down-conversion~\cite{KasamiLin,BST}. The descended family remains superpolynomial and quarter-rate relative to $V$, but its polynomial degree is only about $\lambda n/4$. Thus its effective starting rate relative to the full-field length $n$ is $\lambda/4$, providing the lower-rate anchors needed for padding. Each root placed outside the smaller evaluation set contributes an additional agreement without overlapping the original agreements.

Choosing a smaller evaluation set of appropriate density
$\lambda$ lets padding produce superpolynomial lists at every fixed
rate, with agreement exceeding the rate by a fixed positive amount.
Optimizing $\lambda$ yields the red curve in
Figure~\ref{fig:bounds}.

\paragraph{Beyond full-field evaluation.}
The obstruction survives controlled dense puncturing;
Section~\ref{sec:consequences} gives the quantitative
bound and its primitive and shortening consequences.

\subsection{Further related work}\label{sec:related}
\paragraph{Classical BCH structure.}
Berlekamp's weight enumerators and Kasami--Lin's minimum-weight
results provide historical context~\cite{Berlekamp,KasamiLin}.
Schmidt's rank-refined enumeration is our quantitative input,
while the polynomial descent is related to modern
down-conversion formulations~\cite{Schmidt2020,BST}.

\paragraph{Further capacity-oriented constructions.}
Recent work refines list-decoding guarantees for folded RS and
multiplicity codes, including improved algorithms and relaxed
generalized Singleton bounds~\cite{ChenZhang,AHS}.
Berman--Shany--Tamo instead construct explicit RS subcodes
that achieve capacity by restricting the message space~\cite{BSTsubcodes}.
Neither addresses the unrestricted message space of the canonical
ordinary RS codes studied here.
For ordinary RS codes, random puncturing provides complementary
capacity results under suitable alphabet-size conditions~\cite{AGLZ};
Section~\ref{sec:consequences} compares them with our dense-puncturing obstruction.

\paragraph{Proximity testing.}
RS proximity underlies FRI and related low-degree tests~\cite{FRI2018,DEEPFRI}. Recent work has clarified that ordinary list decoding, correlated agreement, and proximity-gap statements are related but distinct, and that their behavior can depend on field arithmetic~\cite{BCHKS,CS,KKH}. Our theorem concerns ordinary Hamming list size; its relevance to proximity testing is therefore through that combinatorial obstruction rather than through a direct statement about a particular proof system.

\subsection{Organization}\label{sec:organization}
Section~\ref{sec:prelim} introduces the coding-theoretic notation, the BCH--RS correspondence, and anchor families. Section~\ref{sec:anchor} constructs a full-field anchor family from a BCH minimum-weight multiplicity input. Section~\ref{sec:rate-lifting} develops common-root padding. Section~\ref{sec:downconversion} opens the BCH family structurally, proves common-kernel down-conversion, and derives the all-rate theorem. Section~\ref{sec:consequences} collects nearby RS consequences. Section~\ref{sec:future} records a short list of open directions. Longer counting and consequence proofs are deferred to the appendices.

\section{Preliminaries and the BCH--RS correspondence}\label{sec:prelim}

\subsection{Coding-theoretic notation}
For a given extension degree $h\ge1$, write
\[
 K=\F_{2^h},\qquad K^*=K\setminus\{0\}.
\]
We reserve $X$ for a polynomial indeterminate, use lower-case Latin letters such as $p(X)$ and $c(X)$ for polynomials, and use lower-case Greek letters such as $\beta,\gamma\in K$ for field elements.

For a finite evaluation set $U\subseteq K$ and an integer $1\le k\le |U|$, the Reed--Solomon code of dimension $k$ on $U$ is
\[
 \RS_U(k)
 =\bigl\{(p(\beta))_{\beta\in U}:p\in K[X],\ \deg p<k\bigr\}.
\]
Our main theorem concerns the \emph{full-field} code $\RS_K(k)$, of block length $2^h$, dimension $k$, rate $R=k/2^h$, and minimum Hamming distance $2^h-k+1$. When $U=K^*$ we obtain the \emph{primitive} RS code $\RS_{K^*}(k)$ of length $2^h-1$.

For words $u,v\in K^U$, their Hamming distance is the number of coordinates on which they differ. Equivalently, their \emph{agreement} is the number of coordinates on which they are equal. If $y:U\to K$ is a received word and $p\in K[X]$, write
\[
 \operatorname{agr}_U(p,y)
 =\bigl|\{\beta\in U:p(\beta)=y(\beta)\}\bigr|.
\]
Thus relative distance at most $1-R-\varepsilon$ is the same, up to the usual finite-length rounding, as relative agreement at least $R+\varepsilon$. List decoding asks to bound the number of degree-$<k$ polynomials satisfying such an agreement condition.

For a $\F_2$-valued function $f:U\to\F_2$, write
\[
 \wt_U(f)=|\{\beta\in U:f(\beta)=1\}|.
\]
When $U$ is clear we omit the subscript. Every function on a finite set $U\subseteq K$ has a unique interpolating polynomial of degree $<|U|$; for $\F_2$-valued $f$ we denote this polynomial by $c_f(X)$.

\subsection{The extended binary BCH model}
For even $2\le D<2^h$, define the following evaluated-polynomial realization of the extended primitive binary BCH code:
\begin{equation}\label{eq:ebch-definition}
 \operatorname{eBCH}(2^h,D)
 =\left\{(c(\beta))_{\beta\in K}:
    c\in K[X],\ \deg c\le 2^h-D,\ c(K)\subseteq\F_2
   \right\}.
\end{equation}
This is the convention used, for example, by Berman--Shany--Tamo~\cite[Def.~II.3]{BST}. The code has length $2^h$ over $\F_2$. The integer $D$ is its designed-distance parameter: a nonzero polynomial in~\eqref{eq:ebch-definition} has at most $2^h-D$ zeros, so every nonzero codeword has weight at least $D$. The words of weight exactly $D$ will be counted by the weight-enumerator coefficient below.

Write
\[
 A(2^h,D):=
 \bigl|\{u\in\operatorname{eBCH}(2^h,D):\operatorname{wt}(u)=D\}\bigr|
\]
for the weight-$D$ coefficient of the Hamming weight enumerator of
$\operatorname{eBCH}(2^h,D)$.

\subsection{An exact BCH--RS correspondence}
For even $2\le D<2^h$, define
\[
\mathcal B_K(D)
 =\{c\in K[X]:\deg c\le 2^h-D,\ c(K)\subseteq\F_2,\ \wt_K(c)=D\}
\]
and
\[
\mathcal L_K(D)
 =\left\{p\in K[X]:\deg p\le D/2,  \operatorname{agr}_K(p,\beta\mapsto\beta^{2^{h-1}})\ge D\right\}.
\]

\begin{proposition}[Weight-$D$ BCH words are an exact RS list]\label{prop:bijection}
The maps
\[
 c(X)\longmapsto p(X)
 \quad\text{where}\quad
 \frac{X^{2^h}-X}{c(X)}=p(X)^2+X,
\qquad
 p(X)\longmapsto
 c(X)=\frac{X^{2^h}-X}{p(X)^2+X}
\]
are mutually inverse bijections between $\mathcal B_K(D)$ and
$\mathcal L_K(D)$. In fact, every polynomial in $\mathcal L_K(D)$ has
degree exactly $D/2$ and agreement exactly $D$. Consequently,
\[
 A(2^h,D)=|\mathcal L_K(D)|.
\]
\end{proposition}
\begin{proof}
Take $c\in\mathcal B_K(D)$. It vanishes at exactly $2^h-D$ elements of $K$, so the degree bound is tight: $\deg c=2^h-D$. Hence $c$ splits into distinct linear factors over $K$, all among the factors of
\[
 X^{2^h}-X=\prod_{\beta\in K}(X-\beta),
\]
and therefore $c\mid X^{2^h}-X$. Write
\[
 b(X)=\frac{X^{2^h}-X}{c(X)}.
\]
Then $\deg b=D$, and its roots are exactly the $D$ points where $c(\beta)=1$.

Differentiating
\[
 b(X)c(X)=X^{2^h}-X
\]
gives, in characteristic two,
\[
 b'(X)c(X)+b(X)c'(X)=1.
\]
At every root $\beta$ of $b$ we have $c(\beta)=1$, hence $b'(\beta)=1$. Since $b'(X)-1$ has degree at most $D-1$ and at least $D$ roots, $b'(X)=1$ identically.

Write $b(X)=\sum_j b_jX^j$. In characteristic two, differentiation kills the even powers, so $b'(X)=1$ says that the only odd-degree term of $b$ is $X$. Thus $b(X)-X$ contains only even powers. Because squaring is an automorphism of the finite field $K$, its coefficients have unique square roots, and there is a unique $p\in K[X]$ such that
\[
 b(X)=p(X)^2+X.
\]
Since $D$ is even, $\deg p=D/2$. At every root $\beta$ of $b$ we have $p(\beta)^2=\beta$, hence $p(\beta)=\beta^{2^{h-1}}$. Thus $p\in\mathcal L_K(D)$.

Conversely, let $p\in\mathcal L_K(D)$ and set $b=p^2+X$. Then $b'=1$, while $\deg b\le D$ and the agreement condition gives at least $D$ distinct roots in $K$. Hence $\deg b=D$ and these are exactly its roots. Therefore $b\mid X^{2^h}-X$, and
\[
 c(X)=\frac{X^{2^h}-X}{b(X)}
\]
has degree $2^h-D$. At a nonroot of $b$ we have $c(\beta)=0$; at a root, the differentiated identity above gives $c(\beta)=1$ because $b'(\beta)=1$. Hence $c\in\mathcal B_K(D)$. The two constructions are inverse by definition.
\end{proof}

\paragraph{Historical locator-polynomial connection.}
The complementary factor $b(X)=p(X)^2+X$ is a
support-locator polynomial, normalized by $b'\equiv1$. Earlier, Augot, Charpin, and Sendrier studied
locator polynomials of minimum-weight BCH codewords
using symmetric functions and Newton identities~\cite{ACS}.
Here the characteristic-two identity $b'\equiv1$
follows directly from the binary values and the
full-field normalization.

\begin{remark}[Where characteristic two enters]\label{rem:binary-obstruction}
Since $c$ is $\F_2$-valued, it equals one
at every root of the complementary factor $b$.
Differentiating $bc=X^{2^h}-X$ therefore gives
$b'(\beta)=1$ at all its $D$ roots, and the
degree bound implies $b'\equiv1$. Over odd
characteristic, nonzero BCH symbols need not agree,
so the analogous derivative values
$-c(\beta)^{-1}$ need not be constant.
This is the obstruction to this particular
BCH--RS correspondence in odd characteristic.
\end{remark}

\subsection{Subspace polynomials and additive domains}
Let $V\le K$ be a binary linear subspace. Its \emph{subspace polynomial} is
\[
 \ell_V(X)=\prod_{\nu\in V}(X-\nu).
\]
The subspace polynomial $\ell_V$ is \emph{linearized}: it has the
form $\sum_i a_iX^{2^i}$~\cite{BKR}. Hence evaluation of $\ell_V$
is an $\F_2$-linear map $K\to K$ whose kernel is exactly $V$.
The polynomial is monic of degree $|V|$, and its derivative is a
nonzero constant. It follows that $\ell_V(K)$ is an additive subspace. Moreover, every element of this image has exactly $|V|$ preimages in $K$: if $\ell_V(x)=\ell_V(y)$, then $x-y\in V$, so the preimages of any image point form one coset of $V$.

The full-field correspondence has the following additive-domain version.

\begin{proposition}[Additive-domain BCH--RS correspondence]\label{prop:additive}
Let $V\le K$ be an additive subspace and put
$n_0:=|V|$ and $\nu_V:=\ell_V'\in K^*$. For even $2\le D<n_0$, there is a bijection
between
\[
 \left\{g\in K[X]:
   \deg g\le n_0-D,\ g(V)\subseteq\F_2,\ \wt_V(g)=D
 \right\}
\]
and the polynomials $p\in K[X]$ of degree at most $D/2$ satisfying
\[
 \operatorname{agr}_V(p,\beta\mapsto\beta^{2^{h-1}})\ge D.
\]
The maps are characterized by
\[
 \frac{\ell_V(X)}{g(X)}
   =\nu_V\bigl(p(X)^2+X\bigr),
 \qquad
 g(X)=\frac{\ell_V(X)}
 {\nu_V\bigl(p(X)^2+X\bigr)}.
\]
As in Proposition~\ref{prop:bijection}, the degree and agreement
inequalities are equalities.
\end{proposition}
\begin{proof}
As in Proposition~\ref{prop:bijection}, the $n_0-D$
zeros of $g$ saturate the degree bound, so
$g\mid\ell_V$ and $b:=\ell_V/g$ has degree $D$,
with roots exactly where $g=1$. Differentiating
$bg=\ell_V$ at those roots gives $b'(\beta)=\nu_V$.
Because $\deg b'<D$, we have $b'\equiv\nu_V$,
and hence, in characteristic two, uniquely
\[
 b=\nu_V(p^2+X),\qquad \deg p=D/2.
\]
The $D$ roots give the required agreements with
$\beta^{2^{h-1}}$.

Conversely, such a polynomial $p$ yields
$b=\nu_V(p^2+X)$ of degree $D$ with $D$ distinct
roots in $V$ and derivative $\nu_V$. Thus
$b\mid\ell_V$; for $g=\ell_V/b$, the identity
$b'g+bg'=\nu_V$ shows that $g=1$ at those
roots and $g=0$ elsewhere on $V$.
This recovers the prescribed binary weight
and degree, and the constructions are inverse.
\end{proof}

\subsection{Anchor families}
When the binary weight approaches half the evaluation-set size
from below, the correspondence produces messages of degree below
one quarter of that size, with agreement approaching one half. We call such lists \emph{anchor
families}; the padding argument will use these properties without
referring to their BCH origin.

\begin{definition}[Anchor family]\label{def:anchor-family}
Fix an integer $c\ge0$. For each $h>c$, let
$V_h\le(\F_{2^h},+)$ be an additive subspace of
codimension $c$, and put $n_0:=|V_h|=2^{h-c}$.
A sequence $(\mathcal A_h)$ of nonempty sets
$\mathcal A_h\subseteq\F_{2^h}[X]$ is an
\emph{anchor family of codimension $c$}
(or \emph{density $\lambda=2^{-c}$}) if every
$p\in\mathcal A_h$ has degree $<n_0/4$ and
\[
 \min_{p\in\mathcal A_h}
 \frac{\operatorname{agr}_{V_h}(p,X^{2^{h-1}})}{n_0}
 \ge\tfrac12-o(1)\qquad(h\to\infty).
\]
\end{definition}

By Proposition~\ref{prop:additive}, $\F_2$-valued
polynomials of even weight
$D_h=(1/2-o(1))n_0<n_0/2$ and degree at most
$n_0-D_h$ give an anchor family: the corresponding
messages have degree $D_h/2<n_0/4$ and exactly
$D_h$ agreements.

\section{Stage 1: a full-field anchor family from BCH multiplicity}\label{sec:anchor}
By Proposition~\ref{prop:bijection}, the number of
minimum-weight words in a suitable extended binary BCH code
is also the size of an RS list. We use the following
multiplicity bound, proved in Appendix~\ref{app:bch-multiplicity}
using Schmidt's quadratic-form enumeration~\cite{Schmidt2020}.
Berlekamp's weight-enumerator work provides historical
context~\cite{Berlekamp}.

\begin{proposition}[BCH minimum-weight multiplicity]\label{prop:classical-bch}
For all sufficiently large $h$, there is an even integer
$D_h^\star<2^{h-1}$ such that
\[
 \frac{D_h^\star}{2^h}=\frac12-o(1),\qquad
 A(2^h,D_h^\star)\ge2^{h^2/8-O(h)}.
\]
\end{proposition}

\begin{theorem}[Full-field anchor family]\label{thm:anchor}
There is an anchor family $(\mathcal A_h)$ of codimension
$0$, with $V_h=\F_{2^h}$, such that for all sufficiently
large $h$,
\[
 |\mathcal A_h|\ge2^{h^2/8-O(h)}.
\]
Thus the full-field code $\RS_{\F_{2^h}}(2^{h-2})$
has a superpolynomial list at rate $1/4$
and limiting agreement $1/2$.
\end{theorem}
\begin{proof}
Apply Proposition~\ref{prop:bijection} to the
$A(2^h,D_h^\star)$ weight-$D_h^\star$ BCH words supplied by
Proposition~\ref{prop:classical-bch}. This gives distinct messages
of degree $D_h^\star/2<2^{h-2}$, all agreeing with the received function
$\beta\mapsto\beta^{2^{h-1}}$ at exactly
$D_h^\star=(1/2-o(1))2^h$ points. Together they form
the claimed anchor family.
\end{proof}

\section{Stage 2: common-root padding}\label{sec:rate-lifting}
We show how common-root padding lifts an anchor family
of fixed codimension $c$ to a full-field RS code.
The resulting agreement depends on its density
$\lambda=2^{-c}$ and the target rate $R$.

Fix an anchor family $(\mathcal A_h)$ of codimension $c$.
For a given $h\ge c+2$, write $K=\F_{2^h}$,
$n=|K|$, $V=V_h$, and $n_0=|V|=\lambda n$.
Define the received function on the full field by
\[
 \rho(\beta)=\beta^{2^{h-1}}\qquad(\beta\in K).
\]
For $T\subseteq K$, write
$z_T(X)=\prod_{\zeta\in T}(X-\zeta)$.
Every anchor message is a polynomial of degree $<n_0/4$,
representing a word in $\RS_V(n_0/4)$; our lower bound on
$h$ ensures that $n_0/4$ is an integer. To obtain
messages of degree below $k=\lfloor Rn\rfloor$, we multiply
by a polynomial with
\[
 b:=k-\frac{n_0}{4}
\]
distinct roots. For any $T\subseteq K$ with $|T|=b$,
the padded message $z_Tp$ has degree $<k$.
The constraint $\lambda\le4R$ ensures $b\ge0$
for all sufficiently large $h$. When $\lambda=4R$,
$b=0$ exactly.

A padding root creates an automatic agreement after
multiplying both the messages and $\rho$ by $z_T$. If the root lies outside $V$,
no original agreement is lost. If it lies inside $V$,
an agreement may already have been present. The
switch occurs when the $b$ roots no longer fit into
$K\setminus V$. In the limit, the condition is
\[
 R-\frac{\lambda}{4}\le1-\lambda
 \quad\Longleftrightarrow\quad
 R\le1-\frac{3\lambda}{4}.
\]

\begin{proposition}[Padding an anchor family]
\label{prop:asymptotic-padding}
Let $(\mathcal A_h)$ be an anchor family of fixed
codimension $c$ and density $\lambda=2^{-c}$,
and write $M_h=|\mathcal A_h|$.
For any fixed $0<R<1$ with $\lambda\le4R$, define
\begin{equation}\label{eq:padding-agreement}
 \alpha_\lambda(R):=
 \begin{cases}
 R+\dfrac{\lambda}{4},
    &R\le1-\dfrac{3\lambda}{4}
      \quad\text{(disjoint)},\\[4pt]
 \dfrac{1+R}{2}-\dfrac{\lambda}{8},
    &R\ge1-\dfrac{3\lambda}{4}
      \quad\text{(overlapping)}.
 \end{cases}
\end{equation}
The expressions agree at the switching rate.
For every fixed $0<\alpha<\alpha_\lambda(R)$,
all sufficiently large $h$ admit at least $M_h/n$
distinct padded messages of degree $<\lfloor Rn\rfloor$
with agreement fraction at least $\alpha$ against
a common received word on $K$.
\end{proposition}
\begin{proof}
By the anchor-family definition, we may take
$\eta_h\to0$ such that each $p\in\mathcal A_h$
agrees with $\rho$ at no fewer than
$(1/2-\eta_h)n_0$ points of $V$.
Let $k=\lfloor Rn\rfloor$ and $b=k-n_0/4$.
For each padding set $T$ of size $b$,
use $z_T\rho$ as the received word
and $z_Tp$ as the padded messages.
Their distinctness is preserved by multiplication
by the same nonzero polynomial $z_T$.

\emph{Disjoint padding.}
If $b\le n-n_0$, choose $T\subseteq K\setminus V$.
All original anchor agreements survive and the
$b$ roots give additional agreements. Every padded
message therefore agrees at least
\[
 b+(1/2-\eta_h)n_0
\]
times. Dividing by $n$ gives the limiting fraction
\[
 \left(R-\frac{\lambda}{4}\right)
       +\frac{\lambda}{2}
 =R+\frac{\lambda}{4}.
\]
All $M_h$ messages are retained.

\emph{Overlapping padding.}
If $b>n-n_0$, use every point outside $V$ as a
padding root and choose the remaining roots
uniformly from $V$. Equivalently, the $n-b$
unpadded coordinates form a uniformly random
subset of $V$. A uniformly selected unpadded
coordinate is an agreement for a fixed anchor
$p$ with probability
$\operatorname{agr}_V(p,\rho)/n_0\ge1/2-\eta_h$.
Including the $b$ automatic agreements,
the expected number of agreements for each
padded message is at least
\[
 n\mu_h,\qquad
 \mu_h:=\frac bn+
          \left(1-\frac bn\right)
          \left(\frac12-\eta_h\right).
\]
Since $b/n\to R-\lambda/4$, we obtain
\[
 \mu_h\longrightarrow
 \frac{1+R}{2}-\frac{\lambda}{8}.
\]

We choose the \emph{same} random padding set for
all anchors. Averaging shows that some set $T$
gives an average of at least $n\mu_h$ agreements
across the $M_h$ messages. Put
$A=\lfloor n\mu_h\rfloor$. If $G$ messages have at least
$A$ agreements, the rest have at most $A-1$, so
\[
 M_hn\mu_h\le(M_h-G)(A-1)+Gn,
 \qquad
 G\ge\frac{M_h(n\mu_h-A+1)}{n-A+1}
 \ge\frac{M_h}{n}.
\]

At the switching rate $R=1-3\lambda/4$, we have
$b=n-n_0$ exactly, so the disjoint construction
uses every point outside $V$.
The rounding loss $1/n$ in the overlapping case
vanishes, proving the assertion for every fixed
$\alpha<\alpha_\lambda(R)$.
\end{proof}

\begin{corollary}[Padding the full-field anchor family]
\label{cor:anchor-lifting}
For every fixed $1/4\le R<3/4$ and every fixed
$R<\alpha<3/8+R/2$, all sufficiently large $h$
admit a received word for the full-field RS code
over $\F_{2^h}$ of dimension $\lfloor R2^h\rfloor$
with at least $2^{h^2/8-O_{R,\alpha}(h)}$ codewords
at agreement fraction at least $\alpha$.
\end{corollary}
\begin{proof}
Theorem~\ref{thm:anchor} gives anchor families
of codimension $0$ and size $2^{h^2/8-O(h)}$,
so $\lambda=1$. Padding is feasible at every
$R\ge1/4$, including equality. The overlapping
branch of~\eqref{eq:padding-agreement}
(or either branch at $R=1/4$) gives
\[
 \alpha_1(R)=\frac{1+R}{2}-\frac18
             =\frac38+\frac R2.
\]
This exceeds $R$ for $R<3/4$. The loss of
at most a factor $n=2^h$ preserves the
leading $h^2/8$ list-size exponent.
\end{proof}

\section{Stage 3: anchor families on smaller additive domains}\label{sec:downconversion}
Throughout this section, take $h\ge3$, write $K=\F_{2^h}$
and $n=2^h$, and regard $K$ as an $h$-dimensional vector
space over $\F_2$.

To apply padding at every rate, we construct large anchor
families of fixed codimension $c$. Schmidt's quadratic-form
families provide many binary words with large subspaces of
invariant directions. We select a common subspace $W$ for
many such words and use polynomial descent and the additive-domain
BCH--RS correspondence to obtain the anchor family.

\subsection{Quadratic forms and translation invariance}\label{subsec:quadratic-basics}
This subsection is a gentle introduction to quadratic forms in
characteristic two, focusing on the properties needed below.
Our goal is to find many functions that are constant on the cosets
of a common additive subspace, so that they can descend to a smaller
evaluation domain. We first name this property.

\begin{definition}[$W$-invariant function]\label{def:W-invariant}
Let $W\le(K,+)$ be an additive subspace. A function
$f:K\to K$ is \emph{$W$-invariant} if
\[
 f(x+\omega)=f(x)\qquad
 \text{for every }x\in K\text{ and }\omega\in W.
\]
Equivalently, $f$ is constant on each coset of $W$.
This terminology also applies to $\F_2$-valued
functions on $K$.
\end{definition}

Quadratic forms provide a way to identify such invariant directions.

\begin{definition}[Quadratic form and polar form]\label{def:quadratic-polar}
A map $\varphi:K\to\F_2$ is a \emph{quadratic form} if
$\varphi(0)=0$ and
\[
 B_\varphi(x,y):=\varphi(x+y)+\varphi(x)+\varphi(y)
\]
is bilinear over $\F_2$. The bilinear form $B_\varphi$ is the
\emph{polar form} of $\varphi$. Since $B_\varphi(x,x)=0$ for every $x$,
it is alternating.
\end{definition}

The polar form measures the effect of translation. Indeed, for
every $x,\omega\in K$,
\begin{equation}\label{eq:quadratic-translation}
 \varphi(x+\omega)+\varphi(x)
 =\varphi(\omega)+B_\varphi(x,\omega).
\end{equation}
The change is independent of the base point $x$ precisely when
$B_\varphi(x,\omega)=0$ for all $x$. This motivates the
polar radical.

\begin{definition}[Polar radical and polar rank]\label{def:radical-rank}
The \emph{polar radical} of $\varphi$ is
\[
 \operatorname{rad}(B_\varphi)
 :=\{\omega\in K:B_\varphi(x,\omega)=0
       \text{ for every }x\in K\}.
\]
By bilinearity, $\operatorname{rad}(B_\varphi)$ is a linear subspace of $K$.
The \emph{polar rank} is
\[
 \operatorname{rank}(B_\varphi)
 :=h-\dim_{\F_2}\operatorname{rad}(B_\varphi).
\]
The polar form is \emph{nondegenerate} when $\operatorname{rad}(B_\varphi)=\{0\}$:
then no nonzero direction has zero polar interaction with every other
direction.
\end{definition}

For $\omega\in\operatorname{rad}(B_\varphi)$, the translation
identity reduces to
\[
 \varphi(x+\omega)+\varphi(x)=\varphi(\omega)
 \qquad\text{for every }x.
\]
Such directions change $\varphi$ by a constant, but they need not
leave it unchanged: invariance requires $\varphi(\omega)=0$ as well.

To see that the invariant directions form a subspace, observe that
for $u,v\in\operatorname{rad}(B_\varphi)$ we have
$B_\varphi(u,v)=0$, hence
$\varphi(u+v)=\varphi(u)+\varphi(v)$.
Thus $\varphi$ restricts to an $\F_2$-linear functional on its
polar radical, whose kernel is the desired subspace.

\begin{definition}[Quadratic radical and quadratic rank]\label{def:quadratic-rank}
The \emph{quadratic radical} of $\varphi$ is
\[
 \operatorname{rad}(\varphi)
 :=\{\omega\in\operatorname{rad}(B_\varphi):\varphi(\omega)=0\}
 =\ker\bigl(\varphi|_{\operatorname{rad}(B_\varphi)}\bigr).
\]
The \emph{quadratic rank} is
\[
 \operatorname{rank}(\varphi)
 :=h-\dim_{\F_2}\operatorname{rad}(\varphi).
\]
\end{definition}

By the translation identity, the quadratic radical is precisely
the space of invariant directions:
\[
 \varphi\text{ is }W\text{-invariant}
 \quad\Longleftrightarrow\quad
 W\subseteq\operatorname{rad}(\varphi).
\]
Since the restriction of $\varphi$ to its polar radical is a
linear functional, its kernel has codimension zero or one.
Therefore
\begin{equation}\label{eq:quadratic-polar-rank}
 \begin{aligned}
 \operatorname{rank}(\varphi)-\operatorname{rank}(B_\varphi)
 &=\dim_{\F_2}\operatorname{rad}(B_\varphi)
   -\dim_{\F_2}\operatorname{rad}(\varphi)\\
 &\in\{0,1\}.
 \end{aligned}
\end{equation}
The quadratic and polar ranks coincide exactly when
$\varphi$ vanishes on its polar radical.
In particular, a nonzero linear functional has quadratic rank $1$.

\begin{example}[Polar and quadratic radicals]\label{ex:binary-quadratic}
On $U=\F_2^3$, consider
\[
 \varphi_0(x_1,x_2,x_3)=x_1x_2,
 \qquad
 \varphi_1(x_1,x_2,x_3)=x_1x_2+x_3.
\]
They have the same polar form
$B(x,y)=x_1y_2+x_2y_1$ and hence the same polar radical
\[
 \operatorname{rad}(B_{\varphi_0})
 =\operatorname{rad}(B_{\varphi_1})
 =\langle(0,0,1)\rangle.
\]
Their quadratic radicals differ:
\[
 \operatorname{rad}(\varphi_0)=\langle(0,0,1)\rangle,
 \qquad
 \operatorname{rad}(\varphi_1)=\{0\}.
\]
Thus both have polar rank $2$, but their quadratic ranks are
$2$ and $3$, respectively; their Hamming weights are $2$ and $4$.
Translation by $(0,0,1)$ fixes $\varphi_0$ but flips
$\varphi_1$. Adding the linear function $x_3$ to $\varphi_1$
recovers $\varphi_0$ and restores invariance.
\end{example}

The comparison of polar and quadratic ranks becomes useful
through the following structural property.

\begin{proposition}[Evenness of polar rank]\label{prop:even-polar-rank}
Every quadratic form $\varphi:K\to\F_2$ has even polar rank.
\end{proposition}
\begin{proof}
The alternating-form argument is given in
Appendix~\ref{app:quadratic-even-rank}.
\end{proof}

Together with~\eqref{eq:quadratic-polar-rank}, this means that
the parity of quadratic rank records exactly whether $\varphi$
vanishes on its polar radical: the quadratic rank is even when
it does, and odd otherwise. This distinction governs the
Hamming weight of $\varphi$. In particular, odd quadratic rank
forces a balanced word, while even quadratic rank yields a
deviation from half the field size, whose magnitude is determined
by the quadratic rank but whose sign is not.

The following character-sum calculation makes this precise.

\begin{proposition}[Quadratic rank and Hamming weight]\label{prop:quadratic-character-sum}
Let $\varphi:K\to\F_2$ be a quadratic form of quadratic rank $r$.
Then
\[
 \wt_K(\varphi)=
 \begin{cases}
  2^{h-1}\pm2^{h-a-1},&r=2a,\\
  2^{h-1},&r=2a+1.
 \end{cases}
\]
For even $r$, the sign depends on the quadratic form,
not just its quadratic rank.
\end{proposition}
\begin{proof}
Put $S=\sum_{x\in K}(-1)^{\varphi(x)}$. Expanding $S^2$ and
using the bijection $(x,y)\mapsto(x,x+y)$ gives
\[
 \begin{aligned}
 S^2
 &=\sum_{x,y\in K}(-1)^{\varphi(x)+\varphi(y)}\\
 &=\sum_{x,z\in K}(-1)^{\varphi(x)+\varphi(x+z)}\\
 &=\sum_{z\in K}(-1)^{\varphi(z)}
     \sum_{x\in K}(-1)^{B_\varphi(x,z)}\\
 &=2^h\sum_{z\in \operatorname{rad}(B_\varphi)}(-1)^{\varphi(z)}.
 \end{aligned}
\]
The last equality evaluates the inner sum over $x$ in the third
line. For $z\notin \operatorname{rad}(B_\varphi)$, the map $x\mapsto B_\varphi(x,z)$
is a nonzero binary linear functional, so its character sum is
zero; for $z\in \operatorname{rad}(B_\varphi)$, that inner sum is $2^h$.

It remains to evaluate the sum over $z\in \operatorname{rad}(B_\varphi)$ in the last
line. As observed above,
$\varphi|_{\operatorname{rad}(B_\varphi)}$ is linear. If it is nonzero, this sum vanishes
and $S=0$. Otherwise it equals $|\operatorname{rad}(B_\varphi)|=2^{h-2a}$, since
the quadratic rank is then the polar rank $2a$. Hence
$S^2=2^{2h-2a}$. The identity
$\wt_K(\varphi)=(2^h-S)/2$ now gives the stated weights.
\end{proof}

\subsection{Schmidt's trace-quadratic family}\label{subsec:schmidt-family}

We now need many quadratic forms with controlled polar rank and
canonical polynomial degree. Schmidt's trace-quadratic families
provide both features.

Put $m=\lfloor h/2\rfloor$ and fix $1\le a\le m$. Throughout the
construction, write $\delta:=2^{h-a-1}$ and
$D_{h,a}:=2^{h-1}-\delta$ for the full-field BCH weight.
Here $\Tr_{E/\F_2}$ denotes the field trace of a finite
extension $E/\F_2$. We use Schmidt's binary trace-quadratic
family $\mathcal F_{h,a}$:
\[
\begin{aligned}
 h=2m+1:\quad \varphi(X)&=\Tr_{K/\F_2}\!\left(\sum_{i=a}^{m}u_i X^{2^i+1}\right),\\
 h=2m:\quad \varphi(X)&=\sum_{i=a}^{m-1}\Tr_{K/\F_2}\!\left(u_i X^{2^i+1}\right)\\
 &\quad+\Tr_{\F_{2^m}/\F_2}\!\left(u_m X^{2^m+1}\right).
\end{aligned}
\]
The coefficients $u_i$
range independently over the indicated fields. For odd $h$, all $u_i$
belong to $K$; for even $h$, $u_a,\ldots,u_{m-1}\in K$ and
$u_m\in\F_{2^m}$. Each displayed function is quadratic: its
polarization replaces $X^{2^i+1}$ by the bilinear expression
$x^{2^i}y+xy^{2^i}$ inside the trace.
We use two properties of $\mathcal F_{h,a}$:
it is \emph{additive}, as is immediate from the coefficient
description, and every nonzero difference of two forms
has quadratic rank at least $2a+1$ for odd $h$ and at
least $2a$ for even $h$~\cite[Theorems~5.4 and~5.6]{Schmidt2020}.

An affine linear shift can make a quadratic form invariant
under its entire polar radical. We count the shifts
that also attain the required BCH weight.

\begin{lemma}[Affine shifts adapted to a polar radical]\label{lem:radical-shifts}
Let $\varphi:K\to\F_2$ be a quadratic form of polar rank $2a$. Among all affine shifts $L+b$, with $L:K\to\F_2$
linear and $b\in\F_2$, exactly $2^{2a}$ choices make
$f=\varphi+L+b$ $\operatorname{rad}(B_\varphi)$-invariant and have weight $D_{h,a}$.
For each selected linear
part $L$, the quadratic radical $\operatorname{rad}(\varphi+L)$ equals $\operatorname{rad}(B_\varphi)$.
\end{lemma}
\begin{proof}
We first determine which linear parts $L$ make
$f=\varphi+L+b$ $\operatorname{rad}(B_\varphi)$-invariant.
The translation identity shows that
\[
 L|_{\operatorname{rad}(B_\varphi)}=\varphi|_{\operatorname{rad}(B_\varphi)}
\]
is sufficient. It is also necessary: if $f$ is
$\operatorname{rad}(B_\varphi)$-invariant, then for every $\omega\in \operatorname{rad}(B_\varphi)$,
\[
 0=f(\omega)+f(0)=\varphi(\omega)+L(\omega).
\]
As established above, $\varphi$ restricts to a linear functional on
its polar radical. This functional extends to $K$: extend a basis of $\operatorname{rad}(B_\varphi)$ to a basis of $K$ and choose its
values arbitrarily on the added basis vectors. Any two extensions
differ by a linear functional that vanishes on $\operatorname{rad}(B_\varphi)$, equivalently
by a linear functional on the quotient $K/\operatorname{rad}(B_\varphi)$. Since this
quotient has dimension $2a$, there are exactly $2^{2a}$ extensions.

Fix one such extension $L$. The quadratic form $\varphi+L$ has the
same polar rank $2a$ and vanishes on $\operatorname{rad}(B_\varphi)$, so its quadratic
rank is $2a$. Proposition~\ref{prop:quadratic-character-sum}
gives weight $2^{h-1}\pm2^{h-a-1}$. Adding the constant $b$
interchanges zeros and ones, so exactly one choice gives the
smaller weight $D_{h,a}$.

For each selected $L$, the quadratic form $\psi:=\varphi+L$
vanishes on its polar radical $\operatorname{rad}(B_\varphi)$. Therefore
$\operatorname{rad}(\psi)=\operatorname{rad}(B_\varphi)$ by the definition of the quadratic radical.
No other affine shifts satisfy both the invariance and weight requirements.
\end{proof}

To place the selected functions in the extended BCH code, we need
a bound on their canonical polynomial degree as well as their weight.

\begin{lemma}[Degree checkpoint for the trace-quadratic family]\label{lem:trace-degree}
Let $\varphi\in\mathcal F_{h,a}$ and let $f=\varphi+L+b$
for any affine shift $L+b$. If $c_f\in K[X]$ is the canonical
degree-$<2^h$ polynomial representing $f$ on $K$, then
\[
 \deg c_f\le 2^{h-1}+\delta.
\]
If in addition $\wt_K(f)=D_{h,a}$, then
\[
 \deg c_f=2^{h-1}+\delta.
\]
\end{lemma}
\begin{proof}
First consider a quadratic trace term. Expanding gives
\[
 \Tr_{K/\F_2}(uX^{2^i+1})
 =\sum_{j=0}^{h-1}
   u^{2^j}X^{2^{i+j}+2^j},
\]
where exponents are reduced as polynomial functions on $K$. Thus each
resulting exponent has two binary $1$-positions, obtained by cyclically
shifting the positions $0$ and $i$ modulo $h$; their circular
separations are $i$ and $h-i$. For the indices in
$\mathcal F_{h,a}$, both separations are at least $a$. Hence, if the
larger position is $h-1$, the other is at most $h-a-1$, and every such
exponent is at most
\[
 2^{h-1}+2^{h-a-1}=2^{h-1}+\delta.
\]
When $h=2m$, the remaining term
$\Tr_{\F_{2^m}/\F_2}(u_m X^{2^m+1})$ is the trace of
$u_m$ times the relative norm $X^{2^m+1}$ from $K$ to
$\F_{2^m}$. Its expansion has exponents
$2^{m+j}+2^j$ for $0\le j<m$, whose binary positions are
separated by $m$, so the same bound applies.
Finally, every $\F_2$-linear functional on $K$ has the form
$\Tr_{K/\F_2}(vX)$ for some $v\in K$.\footnote{This is the trace description of $\operatorname{RM}(1,h)$~\cite{DGM1970}; alternatively, the $2^h$ distinct maps $x\mapsto\Tr(vx)$ exhaust all $2^h$ binary linear functionals.}
Thus an affine shift $L+b$ has canonical degree at most
$2^{h-1}$ and does not increase the bound.

If $\wt_K(f)=D_{h,a}$, then $c_f$ has
$2^{h-1}+\delta$ zeros in $K$; the root bound forces equality
in the degree bound.
\end{proof}

For a prime power $r$, let
\[
 \gbinom{m}{a}{r}
 :=\prod_{j=0}^{a-1}\frac{r^{m-j}-1}{r^{a-j}-1}
\]
be the Gaussian binomial coefficient, which counts $a$-dimensional
subspaces of $\F_r^m$.

We select classes whose polar rank is exactly $2a$, so that
their polar radicals have the common dimension $h-2a$.
Let $\mathcal Q_{h,a}\subseteq\mathcal F_{h,a}$ denote Schmidt's
class of quadratic rank $2a+1$ when $h=2m+1$, and his elliptic
class of quadratic rank $2a$ when $h=2m$.
Here ``elliptic'' specifies Schmidt's class of quadratic rank $2a$
with weight $2^{h-1}+2^{h-a-1}$.
\begin{proposition}[Schmidt's quadratic-form count]\label{prop:schmidt-form-count}
For $1\le a\le m=\lfloor h/2\rfloor$, the selected class
$\mathcal Q_{h,a}$ has cardinality $C_{h,a}:=|\mathcal Q_{h,a}|$,
where
\begin{equation}\label{eq:Cha}
 C_{h,a}=
 \begin{cases}
  (2^h-1)\gbinom{m}{a}{4},&h=2m+1,\\
  (2^a-1)\gbinom{m}{a}{4},&h=2m.
 \end{cases}
\end{equation}
\end{proposition}
This specializes Schmidt's quadratic-rank distribution
theorems~\cite{Schmidt2020}; the parameter translation
is given in Appendix~\ref{app:schmidt-specialization}.

Define the \emph{shifted Schmidt family}
\[
 \mathcal S_{h,a}:=
 \left\{\varphi+L+b:
 \begin{array}{l}
 \varphi\in\mathcal Q_{h,a},\quad
 L\in\operatorname{Hom}_{\F_2}(K,\F_2),\quad b\in\F_2,\\
 L|_{\operatorname{rad}(B_\varphi)}=\varphi|_{\operatorname{rad}(B_\varphi)},\quad
 \wt_K(\varphi+L+b)=D_{h,a}
 \end{array}\right\}.
\]
Its size is determined in the following proposition.

\begin{proposition}[Structured Schmidt subfamily]\label{prop:enumeration}
The $C_{h,a}$ forms in $\mathcal Q_{h,a}$ have polar rank
$2a$ and pairwise distinct polar forms. For each
$\varphi\in\mathcal Q_{h,a}$, exactly $2^{2a}$ affine
shifts belong to $\mathcal S_{h,a}$; they are
$\operatorname{rad}(B_\varphi)$-invariant, satisfy $\operatorname{rad}(\varphi+L)=\operatorname{rad}(B_\varphi)$,
and have canonical polynomial degree $2^h-D_{h,a}$.
The shifts associated with different $\varphi$ are distinct.
Consequently,
\[
 |\mathcal S_{h,a}|=2^{2a}C_{h,a}=:M_{h,a},
\]
and their canonical polynomials are distinct
weight-$D_{h,a}$ codewords of
$\operatorname{eBCH}(2^h,D_{h,a})$.
\end{proposition}
\begin{proof}
We count the relevant Schmidt classes, identify their polar ranks,
and check that the affine shifts give distinct BCH words.

Proposition~\ref{prop:schmidt-form-count} supplies the
cardinality $C_{h,a}$.
The selected quadratic ranks are $2a+1$ (odd $h$) and $2a$
(even $h$, elliptic). By~\eqref{eq:quadratic-polar-rank} and
Proposition~\ref{prop:even-polar-rank}, both classes have
polar rank $2a$:
for odd $h$, the quadratic and polar ranks differ by one, whereas for even
$h$ they coincide. Thus their polar radicals have dimension
$h-2a$.

The selected forms have distinct polar forms: otherwise their
nonzero difference, still in the additive family $\mathcal F_{h,a}$,
would have zero polar form and hence be a nonzero linear function
of quadratic rank $1$, contradicting Schmidt's lower bound
$2a\ge2$ on quadratic rank stated above.

It remains to pass from the selected quadratic forms to BCH words.
Lemma~\ref{lem:radical-shifts} gives exactly $2^{2a}$ affine
shifts for each selected $\varphi$ that are
$\operatorname{rad}(B_\varphi)$-invariant and have weight $D_{h,a}$. Since
$\mathcal Q_{h,a}\subseteq\mathcal F_{h,a}$, these shifts satisfy
both hypotheses of Lemma~\ref{lem:trace-degree}, which gives
$\deg c_f=2^{h-1}+\delta=2^h-D_{h,a}$.
If two such functions from distinct selected $\varphi$ agreed,
their underlying quadratic forms would differ by a linear function,
contrary to the distinctness of the selected polar forms.
Their degree and weight place their canonical polynomials in
$\operatorname{eBCH}(2^h,D_{h,a})$ by~\eqref{eq:ebch-definition}.
Thus the disjoint contributions from the $C_{h,a}$
forms give $|\mathcal S_{h,a}|=2^{2a}C_{h,a}=M_{h,a}$.

\end{proof}

\subsection{A common quotient gives smaller RS anchors}\label{subsec:common-quotient}

Proposition~\ref{prop:enumeration} supplies many binary words,
each invariant under a large additive subspace, though this
subspace may vary from word to word. Before selecting one
subspace shared by many words, we need to understand the effect
of invariance on their canonical polynomials. In particular,
we need to reduce polynomial degree and Hamming weight without
identifying distinct words.

Fix any additive subspace $W\le K$, and write
\[
 \pi:=\ell_W,\qquad w:=|W|,\qquad
 V:=\pi(K),\qquad n_0:=|V|=2^h/w.
\]
The following lemma establishes the required polynomial descent
for any $W$-invariant $\F_2$-valued function.

\begin{lemma}[Polynomial descent through a subspace quotient]\label{lem:subspace-descent}
Let $f:K\to\F_2$ be $W$-invariant, and let $c_f$
be its canonical polynomial of degree $<2^h$. Then:
\begin{enumerate}
\item \emph{Polynomial descent.}
There is a unique polynomial $f^{\downarrow}_W\in K[X]$ of degree $<n_0$,
taking values in $\F_2$ on $V$, such that $c_f=f^{\downarrow}_W\circ\pi$.
\item \emph{Degree scaling.}
For every $f$, $\deg c_f=w\deg f^{\downarrow}_W$, with the convention
$\deg0=-\infty$ and $w(-\infty)=-\infty$.
\item \emph{Weight scaling.}
The Hamming weights satisfy
$\wt_K(f)=w\,\wt_V(f^{\downarrow}_W)$, and the numbers of zero
positions have the same scaling factor $w$.
\item \emph{One-to-one correspondence.}
The map $f\mapsto f^{\downarrow}_W$ is a bijection between $\F_2$-valued
$W$-invariant functions $K\to\F_2$
and polynomials of degree $<n_0$ taking values in $\F_2$
on $V$. In particular, distinct $c_f$ yield distinct $f^{\downarrow}_W$.
\end{enumerate}
\end{lemma}
\begin{proof}
Polynomial descent and degree scaling here specialize
Berman--Shany--Tamo~\cite[Lemma~A.2]{BST}.
We give a direct proof that also establishes weight scaling
and the one-to-one correspondence.

\emph{(1) Polynomial descent.}
The subspace polynomial $\pi$ is $\F_2$-linear with kernel
$W$, and so induces an isomorphism $K/W\cong V$.
Both $f$ and $\pi$ are $W$-invariant, so each coset
$B=x+W$ determines a point
$(\pi(B),f(B))\in V\times\F_2$, where the two entries
denote their common values on $B$. The isomorphism
shows that the $n_0$ cosets have distinct
$\pi$-values, comprising all of $V$. Hence there is a unique
polynomial $f^{\downarrow}_W\in K[X]$ of degree $<n_0$ interpolating
these $n_0$ points. By construction,
$f(x)=f^{\downarrow}_W(\pi(x))$ for every $x\in K$.
As $\deg\pi=w$ and $wn_0=2^h$, the polynomial
$f^{\downarrow}_W\circ\pi$ has degree $<2^h$, as does $c_f$.
Since they agree throughout $K$, they coincide as
polynomials: $c_f=f^{\downarrow}_W\circ\pi$. The same interpolation
uniqueness shows that no other polynomial of degree $<n_0$
can satisfy this identity.

\emph{(2) Degree scaling.}
For $f^{\downarrow}_W\ne0$, polynomial composition gives
$\deg c_f=(\deg\pi)(\deg f^{\downarrow}_W)=w\deg f^{\downarrow}_W$,
including when $f^{\downarrow}_W$ is a nonzero constant.
If $f^{\downarrow}_W=0$, then $c_f=0$ by part~\textup{(1)},
so the identity holds
with the stated convention.

\emph{(3) Weight scaling.}
Each $\beta\in V$ has exactly $w$ preimages in $K$,
and $f$ takes the value $f^{\downarrow}_W(\beta)$ at all of them.
Counting ones and zeros gives the two scaling identities.

\emph{(4) One-to-one correspondence.}
Conversely, any polynomial $g$ of degree $<n_0$ that is
$\F_2$-valued on $V$ yields a $\F_2$-valued function
$g\circ\pi$ on $K$, which is $W$-invariant.
By uniqueness in part~\textup{(1)},
$(g\circ\pi)^{\downarrow}_W=g$. In the other direction,
part~\textup{(1)} gives $f=f^{\downarrow}_W\circ\pi$
as functions on $K$. These maps are inverse bijections,
and distinct canonical polynomials descend to distinct polynomials.
\end{proof}

We now return to the Schmidt family and seek a subspace $W$
shared by many of its words. Let $m=\lfloor h/2\rfloor$ and
$1\le a\le m$, and recall that
$D_{h,a}=2^{h-1}-2^{h-a-1}$ and
$M_{h,a}=|\mathcal S_{h,a}|=2^{2a}C_{h,a}$.
Each $f=\varphi+L+b\in\mathcal S_{h,a}$ is invariant under
$\operatorname{rad}(B_\varphi)$ of dimension $h-2a$. Selecting a common $W$ inside
many such subspaces will let Lemma~\ref{lem:subspace-descent}
produce distinct polynomials on the same smaller domain $V$.

To obtain an anchor domain of codimension $c$, we quotient
by a $c$-dimensional subspace $W$: the image of $\ell_W$
then has dimension $h-c$. Such a $W$ must fit inside the
$h-2a$-dimensional invariance spaces. We also require the
descended weight to be even so that
Proposition~\ref{prop:additive} applies.

Fix $0\le c\le h-2a$ with
$h-a-c-1\ge1$ and set
\[
 n_0=2^{h-c},\qquad
 D:=D_{h,a}/2^c=(2^a-1)2^{h-a-c-1},
 \qquad
 \Gamma_{h,a,c}:=M_{h,a}
 \frac{\gbinom{h-2a}{c}{2}}{\gbinom{h}{c}{2}}.
\]
The displayed factorization of $D$ makes its positivity,
integrality, and evenness immediate.

\begin{proposition}[A common quotient RS anchor]\label{prop:common-quotient-anchor}
There exists a $c$-dimensional additive subspace $W\le K$
such that
\[
 \bigl|\{f\in\mathcal S_{h,a}:
       f\text{ is }W\text{-invariant}\}\bigr|
 \ge\Gamma_{h,a,c}.
\]
With $V=\ell_W(K)$ of size $n_0$, these functions
yield at least $\Gamma_{h,a,c}$ distinct polynomials
$p_f\in K[X]$ of degree $D/2$, all agreeing with the same
received function
\[
 \rho(\beta)=\beta^{2^{h-1}}\qquad(\beta\in V)
\]
at exactly $D$ evaluation points, all of which belong to $V$.
\end{proposition}
\begin{proof}
Consider the incidence set
\[
 \mathcal I_{h,a,c}:=
 \bigl\{(f,W): f\in\mathcal S_{h,a},\
 W\le_{\F_2}K,\ \dim_{\F_2}W=c,\
 f\text{ is }W\text{-invariant}\bigr\}.
\]
For every $f=\varphi+L+b\in\mathcal S_{h,a}$,
its space of all invariant directions is
$\operatorname{rad}(\varphi+L)=\operatorname{rad}(B_\varphi)$, of dimension $h-2a$.
Hence $f$ is $W$-invariant precisely when
$W\subseteq \operatorname{rad}(B_\varphi)$, and
\[
 |\mathcal I_{h,a,c}|
 =|\mathcal S_{h,a}|\gbinom{h-2a}{c}{2}
 =M_{h,a}\gbinom{h-2a}{c}{2}.
\]
There are $\gbinom{h}{c}{2}$ possible $c$-dimensional
subspaces $W\le K$. Counting the pairs by their second
component therefore gives some $W$ for which at least
\[
 \frac{|\mathcal I_{h,a,c}|}{\gbinom{h}{c}{2}}
 =\Gamma_{h,a,c}
\]
members of $\mathcal S_{h,a}$ are $W$-invariant.

For this fixed $W$, Lemma~\ref{lem:subspace-descent}
gives distinct $\F_2$-valued polynomials
$f^{\downarrow}_W$ on the common domain $V$.
By degree and weight scaling, respectively,
\[
 \begin{aligned}
 \deg f^{\downarrow}_W
 &=\frac{\deg c_f}{2^c}
   =\frac{2^h-D_{h,a}}{2^c}
   =2^{h-c-1}+2^{h-a-c-1}=n_0-D,\\
 \wt_V(f^{\downarrow}_W)
 &=\frac{\wt_K(f)}{2^c}
   =\frac{D_{h,a}}{2^c}
   =2^{h-c-1}-2^{h-a-c-1}=D.
 \end{aligned}
\]
Here $V=\ell_W(K)$ is an additive subspace, because
$\ell_W$ is $\F_2$-linear. Hence its subspace
polynomial $\ell_V$ has nonzero constant derivative,
as required in Proposition~\ref{prop:additive}.
Since $D$ is even and $2\le D<n_0$, that
correspondence yields distinct message polynomials
of degree
\[
 \deg p_f=\frac D2
 =2^{h-c-2}-2^{h-a-c-2},
\]
each with exactly $D$ agreements against
$\rho(\beta)=\beta^{2^{h-1}}$ on $V$.
Both the normalization and evaluation domain depend
only on $V$, so the received word is common to the family.
\end{proof}

\begin{corollary}[Anchor families of prescribed codimension]
\label{cor:quotient-anchor-family}
For every fixed integer $c\ge0$, there is an anchor
family $(\mathcal A_h)$ of codimension $c$ satisfying
\[
 |\mathcal A_h|\ge2^{h^2/8-O_c(h)}
\]
for all sufficiently large $h$.
\end{corollary}
\begin{proof}
Put $m=\lfloor h/2\rfloor$. For fixed $c$ and feasible $a$,
Lemma~\ref{lem:counting-estimates} gives
\[
 \log_2\Gamma_{h,a,c}
 =2a(m-a)-2ac+O_c(h).
\]
The leading quadratic term is maximized near $a=m/2$,
while selecting a common $c$-dimensional subspace costs
only $O_c(h)$ bits. We therefore take $a=\lceil m/2\rceil$.
For fixed $c$ and large $h$, the constraints
$c\le h-2a$ and $h-a-c-1\ge1$ hold, so
Proposition~\ref{prop:common-quotient-anchor} applies.
Let $V_h=\ell_W(K)$ be the resulting common domain,
of codimension $c$. Its messages have degree
$D/2<|V_h|/4$ and agree with the received function at
\[
 D=\frac{|V_h|}{2}-2^{h-a-c-1}
    =\bigl(\tfrac12-o(1)\bigr)|V_h|
\]
coordinates. These polynomials form an anchor
family of codimension $c$. Moreover, this choice gives
\[
 \log_2\Gamma_{h,a,c}=h^2/8-O_c(h).
\]
\end{proof}

\begin{remark}[Relation to classical down-conversion]
Classical BCH down-conversion, used in Kasami--Lin's
minimum-distance work and revisited by Berman--Shany--Tamo,
passes from suitable BCH codewords to lower-weight words
of smaller designed distance~\cite{KasamiLin,BST}.
The identity $c_f=f^{\downarrow}_W\circ\ell_W$ expresses
the corresponding descent in the polynomial model:
$W$-invariance permits passage to a smaller additive
evaluation domain, scaling both weight and degree by
$|W|$. The additional ingredient here is a common $W$
for superpolynomially many words, so their distinctness
and multiplicity survive the descent.
\end{remark}
\subsection{Completion by padding}

Fix a target rate $0<R<1$. A dyadic density
$\lambda=2^{-c}$ is \emph{feasible} if $\lambda/4\le R$:
the anchor messages then fit within the target degree bound.
For such $\lambda$, Corollary~\ref{cor:quotient-anchor-family}
provides an anchor family of codimension $c$ and
size $2^{h^2/8-O_c(h)}$. By
Proposition~\ref{prop:asymptotic-padding}, padding
retains at least a $2^{-h}$ fraction of the family
at every fixed agreement below $\alpha_\lambda(R)$.
The resulting list size is $2^{h^2/8-O_c(h)}$.

This already proves the existence of a fixed agreement
gap at every rate. Indeed, choosing any fixed dyadic
density
\[
 0<\lambda<\min\{4R,\tfrac43(1-R)\}
\]
makes padding both feasible and disjoint, giving
$\alpha_\lambda(R)=R+\lambda/4>R$.
For example, at $R=3/4$, the choice $c=2$
(and hence $\lambda=1/4$) yields limiting agreement
$13/16$ by disjoint padding.
It remains to choose $\lambda$ for the best
agreement bound obtainable from these families.

To see the tradeoff, set $g=1-R$.
The difference between the guaranteed agreement
and the target rate is
\[
 \alpha_\lambda(R)-R=
 \begin{cases}
 \lambda/4,&\lambda\le4g/3,\\
 g/2-\lambda/8,&\lambda\ge4g/3.
 \end{cases}
\]
The gain increases with $\lambda$ while padding
is disjoint, but decreases once padding overlaps
the anchor domain. We therefore compare feasible
dyadic densities around the switching point
$\lambda=4g/3$.

For $0<R<1/4$, feasibility forces
$\lambda\le4R<1$; all these choices lie in the
disjoint regime. We take the largest feasible
dyadic density. Let $j\ge3$ be the unique
integer with $2^{-j}\le R<2^{-j+1}$, and set
$\lambda_*(R)=2^{-j+2}$.

For $1/4\le R<1$, every dyadic density
$\lambda\le1$ is feasible. Let $j\ge1$ be the
unique integer with $2^{-j}\le1-R<2^{-j+1}$.
The comparison in
Appendix~\ref{app:optimal-dyadic-choice}
shows that $\lambda_*(R)=2^{-j+1}$ maximizes
the gain. The chosen density is disjoint or
overlapping according to which side of
$1-R=3\lambda_*(R)/4$ the rate lies.

Substituting these choices into
\eqref{eq:padding-agreement} yields
\begin{equation}\label{eq:optimized-agreement-envelope}
 \alpha_*(R):=
 \begin{cases}
 R+2^{-j},&0<R<1/4,\\[3pt]
 R+2^{-j-1},
    &R\ge1/4,\quad 1-R\ge3\cdot2^{-j-1},\\[6pt]
 \dfrac{1+R}{2}-2^{-j-2},
    &R\ge1/4,\quad 1-R\le3\cdot2^{-j-1}.
 \end{cases}
\end{equation}
The two high-rate branches agree at
$R=1-3\cdot2^{-j-1}$. Within each dyadic
interval the switch changes the padding regime,
not the optimal quotient density.

\begin{proposition}[Optimality of the dyadic agreement envelope]
\label{prop:optimal-agreement-envelope}
For every $0<R<1$, the density $\lambda_*(R)$ defined above
maximizes $\alpha_\lambda(R)$ over all feasible dyadic choices
$\lambda=2^{-c}\le4R$. In particular,
\[
 \alpha_*(R)
 =\alpha_{\lambda_*(R)}(R)
 =\max_{c\ge0:\,2^{-c}\le4R}\alpha_{2^{-c}}(R).
\]
\end{proposition}
The proof is given in Appendix~\ref{app:optimal-dyadic-choice}.

\begin{theorem}[All-rate full-field RS lower bound]\label{thm:allrate}
Fix $0<R<1$ and let $\alpha_*(R)$ be defined by
\eqref{eq:optimized-agreement-envelope}. For every fixed
$R<\alpha<\alpha_*(R)$, there exists $h_0=h_0(R,\alpha)$
such that all $h\ge h_0$ admit a received word for the
full-field RS code over $\F_{2^h}$ of dimension
$\lfloor R2^h\rfloor$ with at least
\[
 2^{h^2/8-O_R(h)}
\]
codewords at agreement fraction at least $\alpha$.
\end{theorem}
\begin{proof}
Fix $R$ and $\alpha$ as in the statement, and choose the
optimal dyadic density $\lambda=\lambda_*(R)=2^{-c}\le4R$.
Corollary~\ref{cor:quotient-anchor-family} supplies
$h_1=h_1(R)$ such that every $h\ge h_1$ admits an anchor
family on a common additive domain of density $\lambda$
and size at least $2^{h^2/8-O_c(h)}$.
Since $\alpha<\alpha_\lambda(R)=\alpha_*(R)$,
Proposition~\ref{prop:asymptotic-padding} supplies
$h_2=h_2(R,\alpha)$ such that, for every $h\ge h_2$,
padding retains at least a $2^{-h}$ fraction of this family
at agreement fraction at least $\alpha$.
Taking $h_0=\max\{h_1,h_2\}$ and using that $c$ depends
only on $R$ gives $2^{h^2/8-O_R(h)}$ messages
for every $h\ge h_0$, as claimed.
\end{proof}

The boundary $\alpha_*(R)$ is a limiting construction guarantee:
the result applies at every fixed agreement strictly below it, not
necessarily at equality. For $1/4\le R<1$, Appendix~\ref{app:optimal-dyadic-choice}
also gives the simple bound
\[
 \alpha_*(R)-R\ge\frac{1-R}{4}.
\]

\section{Consequences}\label{sec:consequences}

The full-field obstruction also survives deletion of a positive
fraction of coordinates, uniformly over the deleted set.

\begin{corollary}[Robustness under arbitrary dense puncturing]\label{cor:densepuncture}
Fix $0<r<1$, put $G_r:=\alpha_*(r)-r>0$, and choose
$0<\theta<\min\{G_r,1-r\}$. For every fixed
$0<\xi<(G_r-\theta)/(1-\theta)$
and all sufficiently large $h$, if at most $\theta2^h$
coordinates are deleted from the full-field code
$\RS_{\F_{2^h}}(\lfloor r2^h\rfloor)$, the punctured
code has at least $2^{h^2/8-O_r(h)}$ distinct codewords
agreeing with one common received word on a fraction
of coordinates exceeding the punctured code's rate by at least
\[
 \frac{G_r-\theta}{1-\theta}-\xi,
\]
uniformly over the deleted set.
\end{corollary}
Proof in Appendix~\ref{app:dense-puncturing}.

\begin{remark}[Primitive and shortened RS codes]
The primitive RS family is a special case of
Corollary~\ref{cor:densepuncture}. For any target rate
$0<R<1$, choose a source rate $r<R$ sufficiently close to
$R$ that $\alpha_*(r)>R$, and then choose $\theta,\xi>0$
small enough that the surviving agreement exceeds $R$.
Deleting the coordinate $0$ and using
$\lfloor r2^h\rfloor\le\lfloor R(2^h-1)\rfloor$
for large $h$ gives a superpolynomial list in the
primitive code at rate $R$.

The same corollary transfers to arbitrary dense
\emph{shortenings}. Write $K=\F_{2^h}$ and
$k=\lfloor r2^h\rfloor$. For $U\subseteq K$ with
$t=|U|\le\theta|K|$, first apply the corollary to
$\RS_K(k)$ punctured on $U$. Since $k+t\le |K|$
for large $h$, that punctured code is monomially
equivalent to the shortening of $\RS_K(k+t)$ on $U$:
after deleting $U$, the latter consists of evaluations
of $z_U p$ with $\deg p<k$.
Dividing the surviving coordinates by the nonzero
values of $z_U$ identifies this code with
$\RS_{K\setminus U}(k)$, preserving Hamming
agreements and list sizes. Hence the corollary's
puncturing bound transfers to these shortened
codes.
\end{remark}

This dense negative regime should be compared with two positive results. Over prime fields, Jeronimo establishes capacity-approaching $n^{O_\varepsilon(1)}$ list size on every prescribed evaluation set for every $q\ge n$~\cite{Jeronimo}. The BCPZZ theorem gives an all-rate algorithmic bound under a sufficiently large constant-factor field-size hypothesis~\cite{BCPZZ}. Over arbitrary prime powers, random puncturings achieve capacity with high probability under the characteristic-independent sufficient condition
\[
 q\ge n+k\,2^{O(1/\varepsilon^2)}
\]
of Alrabiah et al.~\cite{AGLZ}. Thus for binary extensions there is an unresolved intermediate range of the alphabet expansion ratio $q/n$ between our dense negative regime and the sufficiently expanded random positive regime; these are sufficient regions, not sharp thresholds.

\section{Further directions}\label{sec:future}

\paragraph{Sharpening the agreement boundary.}
A first quantitative question is how close fixed-gap superpolynomial lists can approach the Johnson agreement $\sqrt R$, uniformly for all sufficiently large $h$. At $R=1/2$ the present envelope reaches $5/8$, versus $1/\sqrt2$ at Johnson. Improving $5/8$ is a concrete benchmark; a superpolynomial lower bound strictly above $\sqrt R$ is impossible because Guruswami--Sudan gives a polynomial list bound there.

\paragraph{The unresolved binary puncturing regime.}
How far can the fixed-gap obstruction persist as the
evaluation-set density $n/q$ decreases? Conversely,
how close to $q/n=1$ can random evaluation sets over
binary extension fields already achieve capacity?
Section~\ref{sec:consequences} compares the presently
known dense negative and sparser random positive regions.

\paragraph{Varying the characteristic and extension degree.}
For list size, the prime-field picture is partly a large-characteristic phenomenon. Jeronimo's geometric list bound extends to finite fields of characteristic
\[
 p>\max\{k-1,B_\varepsilon\},
\]
for a constant $B_\varepsilon$ depending only on the fixed capacity slack~\cite{Jeronimo}. At fixed rate $R$, this includes arbitrary evaluation sets whenever $p>k=Rn+O(1)$ (and $p>B_\varepsilon$), so the characteristic may itself scale linearly with the block length.

This condition also marks where that argument stops applying. For full-field RS over $K=\F_{p^m}$ with fixed $m\ge2$, one has $n=q=p^m$ and $k=\Theta(q)$, so $p>k$ fails asymptotically as $p\to\infty$. The first natural test case is therefore the quadratic-extension family $\F_{p^2}$.

At the opposite extreme, our theorem gives a fixed-gap superpolynomial lower bound for full-field RS over $\F_{2^h}$ for all large $h$, while the present proof does not extend directly to odd characteristic (cf.~Remark~\ref{rem:binary-obstruction}). This leaves two basic families to understand: fixed extension degree $m\ge2$ with $p\to\infty$, and fixed characteristic $p$ with $m\to\infty$. Resolving even the quadratic-extension case, or the case of a fixed odd prime characteristic, would help determine whether the binary obstruction reflects a broader extension-field phenomenon or something genuinely special to characteristic two.

\section*{Acknowledgment}
\textbf{AI Methodology.} OpenAI's GPT-6 Astra was used extensively and was \textbf{responsible for essentially all of the mathematical ideas underlying the main result and its proof}. AI-generated or AI-assisted content appears throughout Sections~1--7 and Appendices~A--B, including approach generation, proof development, exposition, verification, literature and provenance research, and manuscript preparation. The author initiated and directed the project and, after earlier AI-generated approaches stalled, suggested seeking alternative low-rate anchors. GPT-6 Astra then developed the BCH--RS correspondence and quarter-rate anchor, the common-kernel quotient construction, and common-root padding used in the final proof. The author substantially contributed to understanding, selecting, organizing, presenting, and attributing the mathematics, and to positioning the result relative to prior work. In particular, the author introduced the manuscript's organizing definitions of anchor families and $W$-invariant functions, and guided revisions that simplified parameter bookkeeping and improved the exposition across the three stages. GPT-6 Astra also helped rederive and challenge arguments, check parameters and hidden hypotheses, and conduct referee-style audits that led to discarding earlier approaches. The author reviewed the final arguments and sources and accepts responsibility for the paper.

\clearpage
\appendix

\section{Enumerator specialization and exact counts}\label{app:count}

\subsection{Even rank of alternating polar forms}\label{app:quadratic-even-rank}
\begin{proof}[Proof of Proposition~\ref{prop:even-polar-rank}]
Let $B$ be an alternating bilinear form on a binary vector space.
If $B=0$, its rank is zero. Otherwise choose $x,y$ with $B(x,y)=1$
and put $W=\langle x,y\rangle$. The restriction $B|_W$ is
nondegenerate: if $a x+b y$ is orthogonal to both $x$ and $y$, then
$B(ax+by,x)=b$ and $B(ax+by,y)=a$, so $a=b=0$.
For any vector $v$, the vector
\[
 v+B(v,y)x+B(v,x)y
\]
is orthogonal to $W$. Hence the ambient space decomposes as
$W\oplus W^\perp$, and
\[
 \operatorname{rank} B
 =2+\operatorname{rank}(B|_{W^\perp}).
\]
Iterating proves that the rank is even.
\end{proof}

\subsection{Binary BCH enumerator specialization}
\label{app:schmidt-specialization}
We first specialize Schmidt's quadratic-form enumerators
to prove Proposition~\ref{prop:schmidt-form-count}.
The affine-shift and polynomial-degree arguments in Section~5
then supply the full-field BCH words needed for
Proposition~\ref{prop:classical-bch}.

\begin{proof}[Proof of Proposition~\ref{prop:schmidt-form-count}]
For the additive family $\mathcal F_{h,a}$ and a specified
class $\mathcal C$ of quadratic forms, Schmidt's normalized
inner-distribution coefficient is
\[
 \frac{1}{|\mathcal F_{h,a}|}
 \left|\{(\varphi,\varphi')\in\mathcal F_{h,a}^2:
                 \varphi-\varphi'\in\mathcal C\}\right|.
\]
Every fixed difference occurs exactly $|\mathcal F_{h,a}|$
times, so the coefficient counts individual forms of
$\mathcal F_{h,a}$ in $\mathcal C$.
For odd $h=2m+1$, set $q=2$, $d=2a+1$ in Schmidt's
Theorem~5.4, and $s=a+1$ in Theorem~4.7.
For even $h=2m$, set $q=2$ and Schmidt's parameter
$\delta_{\mathrm{Sch}}=a$ in Theorem~5.6, and take $s=a$,
$\tau=-1$ (the elliptic class) in Theorem~4.9.
The Gaussian coefficients use base $q^2=4$, giving
the two cases of~\eqref{eq:Cha}.
\end{proof}

\subsection{Common-quotient multiplicity estimate}
\label{app:counting-estimates}
\begin{lemma}[Common-quotient multiplicity estimate]
\label{lem:counting-estimates}
Fix an integer $c\ge0$ and put $m=\lfloor h/2\rfloor$.
Uniformly for $1\le a\le m$ satisfying
$c\le h-2a$ and $h-a-c-1\ge1$,
\[
 \log_2\Gamma_{h,a,c}
 =2a(m-a)-2ac+O_c(h).
\]
\end{lemma}
\begin{proof}
We first estimate the Gaussian coefficients used below.
For any prime power $r$ and integers $m\ge0$, $0\le a\le m$,
\[
 \gbinom{m}{a}{r}
 =r^{a(m-a)}
 \prod_{j=0}^{a-1}
 \frac{1-r^{-(m-j)}}{1-r^{-(a-j)}}.
\]
Each ratio is at least one since $m-j\ge a-j$,
while the product is at most
\[
 C_r:=\prod_{t=1}^{\infty}(1-r^{-t})^{-1}<\infty.
\]
The infinite product converges because
$\sum_{t\ge1}r^{-t}<\infty$. Consequently,
\[
 \log_r\gbinom{m}{a}{r}=a(m-a)+O_r(1),
\]
uniformly in $m$ and $a$.

Next, Proposition~\ref{prop:schmidt-form-count} and the
definition of $M_{h,a}$ give
\[
 M_{h,a}=2^{2a}(2^{t_h}-1)\gbinom{m}{a}{4},
 \qquad
 t_h=\begin{cases}
 h,&h=2m+1,\\
 a,&h=2m.
 \end{cases}
\]
Using the Gaussian estimate with $r=4$ and
$2^{t_h-1}\le 2^{t_h}-1\le 2^{t_h}$ yields
\[
 \log_2 M_{h,a}
 =2a(m-a)+2a+t_h+O(1)
 =2a(m-a)+O(h).
\]

Finally, the definition of $\Gamma_{h,a,c}$ gives
\[
 \frac{\Gamma_{h,a,c}}{M_{h,a}}
 =\frac{\gbinom{h-2a}{c}{2}}{\gbinom{h}{c}{2}}
 =2^{-2ac}P_{h,a,c},
 \qquad
 P_{h,a,c}:=\prod_{j=0}^{c-1}
 \frac{1-2^{-(h-2a-j)}}{1-2^{-(h-j)}}.
\]
For $c\le h-2a$, every factor of $P_{h,a,c}$
is at most one and at least $1-2^{-(c-j)}$.
Consequently,
\[
 0<\prod_{t=1}^{c}(1-2^{-t})
 \le P_{h,a,c}\le1.
\]
The empty product is one for $c=0$. Since $c$
is fixed, taking logarithms gives
\[
 \log_2\frac{\Gamma_{h,a,c}}{M_{h,a}}
 =-2ac+O_c(1).
\]
Combining the shifted-family and quotient estimates
proves the lemma.
\end{proof}

\subsection{BCH minimum-weight multiplicity}\label{app:bch-multiplicity}

\begin{proof}[Proof of Proposition~\ref{prop:classical-bch}]
Put $m=\lfloor h/2\rfloor$. For every $1\le a\le m$, Proposition~\ref{prop:enumeration}
gives $A(2^h,D_{h,a})\ge M_{h,a}$.
Taking $c=0$ in Lemma~\ref{lem:counting-estimates},
so that $\Gamma_{h,a,0}=M_{h,a}$, gives
\[
 \log_2 M_{h,a}=2a(m-a)+O(h).
\]
Choose $a=\lceil m/2\rceil$, which maximizes
the quadratic term up to rounding, and put
$D_h^\star:=D_{h,a}$. Then
\[
 \log_2 M_{h,a}=\frac{h^2}{8}+O(h),
 \qquad
 \frac{D_h^\star}{2^h}=\frac12-2^{-a-1}=\frac12-o(1).
\]
For sufficiently large $h$,
$D_h^\star=2^{h-1}-2^{h-a-1}$ is even
and strictly smaller than $2^{h-1}$.
\end{proof}

\section{Auxiliary counts and consequence proofs}\label{app:consequence-proofs}
\subsection{Hamming-ball averaging near capacity}
\label{app:hamming-ball-count}
Fix $R\in(0,1)$ and positive constants $\eta,\gamma$ with
$\eta+\gamma<h_2(R)$, where $h_2$ is binary entropy.
Let $C$ be any $[n,k]_q$ code with
$k=\lfloor Rn\rfloor$ and
$n\le q\le2^{\gamma n}$, and put
$s=\lceil\eta n/\log_2q\rceil$.
Here $y$ is uniform in $\F_q^n$ and $B(y,t)$ denotes
the Hamming ball of radius $t$ about $y$.
Counting one Hamming sphere around each codeword gives
\[
 \begin{aligned}
 \mathbb E_y|C\cap B(y,n-k-s)|
 &\ge\binom{n}{k+s}q^{-s}
          (1-q^{-1})^{n-k-s}\\
 &\ge2^{(h_2(R)-\eta-\gamma-o(1))n}.
 \end{aligned}
\]
Here $s/n=o(1)$, so the binomial coefficient has
exponent $h_2(R)n-o(n)$, while
$q^{-s}\ge2^{-\eta n-\log_2q}$
accounts for the $\eta+\gamma$ loss.
Since $q\ge n$, the remaining factor
$(1-q^{-1})^{n-k-s}$ is $2^{-O(1)}$.
The lower bound is exponential, and the corresponding
agreement fraction is
\[
 \frac{k+s}{n}=\frac{k}{n}
        +\Theta\!\left(\frac1{\log q}\right).
\]
Using $k/n$ instead of $R$ retains the floor correction
even when $\log q=\Theta(n)$.

\subsection{Optimal density in dyadic rate ranges}
\label{app:optimal-dyadic-choice}
\begin{proof}[Proof of Proposition~\ref{prop:optimal-agreement-envelope}]
Put $g=1-R$ and $G_\lambda(R)=\alpha_\lambda(R)-R$.
Section~\ref{sec:rate-lifting} gives
\[
 G_\lambda(R)=
 \begin{cases}
 \lambda/4,&\lambda\le4g/3,\\
 g/2-\lambda/8,&\lambda\ge4g/3.
 \end{cases}
\]
Thus the gain increases with $\lambda$ in the disjoint
regime and decreases in the overlapping regime.

First suppose $0<R<1/4$. Then $4R<4g/3$,
so every feasible density $\lambda\le4R$ is
disjoint. Write $2^{-j}\le R<2^{-j+1}$ with
$j\ge3$. The largest feasible dyadic density is
$\lambda_*=2^{-j+2}$, and its agreement is
\[
 \alpha_*(R)=R+\lambda_*/4=R+2^{-j}.
\]

Now suppose $R\ge1/4$ and write
$2^{-j}\le g<2^{-j+1}$ with $j\ge1$.
All dyadic densities at most one are feasible.
Set $d=2^{-j+1}$, so
\[
 d/2\le g<d,\qquad d\le1.
\]
Every other dyadic density is at most $d/2$
or at least $2d$. If $\lambda\le d/2$,
then $\lambda\le g<4g/3$, so it is disjoint
and $G_\lambda(R)\le d/8$.
If $d$ is disjoint, $G_d(R)=d/4$; if $d$
is overlapping, then
\[
 G_d(R)=g/2-d/8\ge d/8.
\]
Thus no smaller density improves on $d$
(equality is possible when $g=d/2$).

If $\lambda\ge2d$, then $\lambda>4g/3$,
so it is overlapping, with
\[
 G_\lambda(R)=g/2-\lambda/8
 \le g/2-d/4.
\]
For overlapping $d$ this is strictly smaller
than $G_d(R)=g/2-d/8$.
For disjoint $d$, it is smaller than
$d/4=G_d(R)$, since $g<d$.
Hence $d$ is also optimal among all feasible
dyadic densities.

For this fixed optimal density, the switch between
disjoint and overlapping padding occurs at
\[
 g=3d/4=3\cdot2^{-j-1},
 \qquad R=1-3\cdot2^{-j-1}.
\]
For $j\ge2$, this lies strictly inside the dyadic
interval $1-2^{-j+1}<R\le1-2^{-j}$.
For $j=1$, it occurs at the endpoint $R=1/4$.
Consequently the two affine pieces within a high-rate
dyadic interval come from the padding-regime transition;
changes of the optimal density at dyadic endpoints are a
different phenomenon. Substituting $d=2^{-j+1}$
into the two padding formulas yields
\eqref{eq:optimized-agreement-envelope}.
At $g=d/2$ an adjacent dyadic density ties,
which justifies the half-open dyadic intervals.
\end{proof}

The feasibility thresholds $R=2^{-c}/4$ introduce new dyadic
densities and hence jumps in the constructed envelope; in particular,
its left limit at $R=1/4$ is $3/8$, whereas
$\alpha_*(1/4)=1/2$.

For $R\ge1/4$, the same comparison yields
\[
 \frac{1-R}{4}\le\alpha_*(R)-R
 \le\frac{1-R}{3}.
\]
The upper bound comes from the continuous
switch $\lambda=4g/3$; the lower bound
follows from $d\ge g$ in the disjoint regime
and $d\le2g$ in the overlapping regime
(with equality only at a dyadic endpoint).

\subsection{Arbitrary dense puncturing}\label{app:dense-puncturing}
\begin{proof}[Proof of Corollary~\ref{cor:densepuncture}]
Write $G_r=\alpha_*(r)-r$, as in the corollary.
Use Theorem~\ref{thm:allrate} at the fixed source agreement
$r+G_r-\eta$, where $0<\eta<G_r-\theta$.
Delete $t\le\theta2^h$ coordinates from the full-field code.
Every message loses at most $t$ agreements, while its degree
remains below $k=\lfloor r2^h\rfloor$ and the punctured
length $2^h-t$ eventually exceeds $k$ because
$\theta<1-r$. Thus distinctness survives, and the new
agreement-minus-rate gap is at least
\[
 \frac{(r+G_r-\eta)2^h-t-k}{2^h-t}
 \ge\frac{G_r-\eta-\theta}{1-\theta}-o(1).
\]
Choose $\eta<\min\{G_r-\theta,(1-\theta)\xi/2\}$
and then take $h$ sufficiently large.
\end{proof}


\begin{thebibliography}{99}
\setlength{\itemsep}{2pt plus 1pt minus 1pt}
\emergencystretch=2em

\bibitem{AGLZ} O.~Alrabiah, Z.~Guo, V.~Guruswami, R.~Li, and Z.~Zhang,
``Random Reed--Solomon codes achieve list-decoding capacity with linear-sized alphabets,''
\emph{Advances in Combinatorics}, 2025:8, 2025.
\url{https://doi.org/10.19086/aic.2025.8}.

\bibitem{AHS} V.~Ashvinkumar, M.~Habib, and S.~Srivastava,
``Algorithmic improvements to list decoding of folded Reed--Solomon codes,''
\emph{Proc. ACM--SIAM Symposium on Discrete Algorithms (SODA)}, pp.~880--898, 2026.
\url{https://doi.org/10.1137/1.9781611978971.35}.

\bibitem{FRI2018} E.~Ben-Sasson, I.~Bentov, Y.~Horesh, and M.~Riabzev,
``Fast Reed--Solomon interactive oracle proofs of proximity,''
\emph{Proc. ICALP}, LIPIcs 107, article 14, 2018.
\url{https://doi.org/10.4230/LIPIcs.ICALP.2018.14}.

\bibitem{BCHKS} E.~Ben-Sasson, D.~Carmon, U.~Hab\"ock, S.~Kopparty, and S.~Saraf,
``On proximity gaps of Reed--Solomon codes,''
\emph{Proc. ACM Symposium on Theory of Computing (STOC)}, pp.~1157--1167, 2026.
\url{https://doi.org/10.1145/3798129.3800827}.

\bibitem{DEEPFRI} E.~Ben-Sasson, L.~Goldberg, S.~Kopparty, and S.~Saraf,
``DEEP-FRI: Sampling outside the box improves soundness,''
\emph{Proc. ITCS}, LIPIcs 151, article 5, 2020.
\url{https://doi.org/10.4230/LIPIcs.ITCS.2020.5}.

\bibitem{BKR} E.~Ben-Sasson, S.~Kopparty, and J.~Radhakrishnan,
``Subspace polynomials and limits to list decoding of Reed--Solomon codes,''
\emph{IEEE Trans. Inform. Theory}, 56(1):113--120, 2010.
\url{https://doi.org/10.1109/TIT.2009.2034780}.

\bibitem{ACS} D.~Augot, P.~Charpin, and N.~Sendrier,
``Studying the locator polynomials of minimum weight codewords of BCH codes,''
\emph{IEEE Trans. Inform. Theory}, 38(3):960--973, 1992.
\url{https://doi.org/10.1109/18.135638}.

\bibitem{Berlekamp} E.~R.~Berlekamp,
``The weight enumerators for certain subcodes of the second order binary Reed--Muller codes,''
\emph{Information and Control}, 17(5):485--500, 1970.
\url{https://doi.org/10.1016/S0019-9958(70)90399-2}.

\bibitem{DGM1970} P.~Delsarte, J.-M.~Goethals, and F.~J.~MacWilliams,
``On generalized Reed--Muller codes and their relatives,''
\emph{Information and Control}, 16(5):403--442, 1970.
\url{https://doi.org/10.1016/S0019-9958(70)90214-7}.

\bibitem{KasamiLin} T.~Kasami and S.~Lin,
``Some results on the minimum weight of primitive BCH codes,''
\emph{IEEE Trans. Inform. Theory}, 18(6):824--825, 1972.
\url{https://doi.org/10.1109/TIT.1972.1054903}.

\bibitem{BST} A.~Berman, Y.~Shany, and I.~Tamo,
``Efficient algorithms for constructing minimum-weight codewords in some extended binary BCH codes,''
\emph{IEEE Trans. Inform. Theory}, 70(11):7673--7689, 2024.
\url{https://doi.org/10.1109/TIT.2024.3465218}.

\bibitem{BSTsubcodes} A.~Berman, Y.~Shany, and I.~Tamo,
``Explicit subcodes of Reed--Solomon codes that efficiently achieve list decoding capacity,''
\emph{IEEE Trans. Inform. Theory}, 71(8):5898--5911, 2025.
\url{https://doi.org/10.1109/TIT.2025.3575899}.

\bibitem{BGM} J.~Brakensiek, S.~Gopi, and V.~Makam,
``Generic Reed--Solomon codes achieve list-decoding capacity,''
\emph{SIAM J. Comput.}, special section STOC 2023, pp.~STOC23-118--STOC23-154, 2024.
\url{https://doi.org/10.1137/23M1598064}.

\bibitem{BCPZZ} J.~Brakensiek, Y.~Chen, A.~Putterman, Z.~Zhang, and K.~Z.~Zheng,
``Algorithmic list decoding of Reed--Solomon codes up to capacity,''
arXiv:2609.08005, September 2026.
\url{https://arxiv.org/abs/2609.08005}.

\bibitem{ChenZhang} Y.~Chen and Z.~Zhang,
``Explicit folded Reed--Solomon and multiplicity codes achieve relaxed generalized Singleton bounds,''
\emph{Proc. ACM Symposium on Theory of Computing (STOC)}, pp.~1--12, 2025.
\url{https://doi.org/10.1145/3717823.3718114}.

\bibitem{CS} E.~Crites and A.~Stewart,
``On Reed--Solomon proximity gaps conjectures,''
In \emph{CRYPTO 2026}, Part IV, LNCS 16803, pp.~128--156, Springer, 2026.
\url{https://doi.org/10.1007/978-3-032-35398-6_5}.

\bibitem{GR} V.~Guruswami and A.~Rudra,
``Explicit codes achieving list decoding capacity: Error-correction with optimal redundancy,''
\emph{IEEE Trans. Inform. Theory}, 54(1):135--150, 2008.
\url{https://doi.org/10.1109/TIT.2007.911222}.

\bibitem{GRlimits} V.~Guruswami and A.~Rudra,
``Limits to list decoding Reed--Solomon codes,''
\emph{IEEE Trans. Inform. Theory}, 52(8):3642--3649, 2006.
\url{https://doi.org/10.1109/TIT.2006.878164}.

\bibitem{GuruswamiSudan} V.~Guruswami and M.~Sudan,
``Improved decoding of Reed--Solomon and algebraic-geometry codes,''
\emph{IEEE Trans. Inform. Theory}, 45(6):1757--1767, 1999.
\url{https://doi.org/10.1109/18.782097}.

\bibitem{Jeronimo} F.~G.~Jeronimo,
``Algorithmic list decoding at capacity and optimal proximity gaps for Reed--Solomon codes,''
arXiv:2609.05870, September 2026.
\url{https://arxiv.org/abs/2609.05870}.

\bibitem{KKH} D.~Krachun, S.~Kazanin, and U.~Hab\"ock,
``Failure of proximity gaps close to capacity,''
IACR ePrint 2026/782, 2026.
\url{https://eprint.iacr.org/2026/782}.

\bibitem{KRSW} S.~Kopparty, N.~Ron-Zewi, S.~Saraf, and M.~Wootters,
``Improved list decoding of folded Reed--Solomon and multiplicity codes,''
\emph{SIAM J. Comput.}, 52(3):794--840, 2023.
\url{https://doi.org/10.1137/20M1370215}.

\bibitem{PV} F.~Parvaresh and A.~Vardy,
``Correcting errors beyond the Guruswami--Sudan radius in polynomial time,''
\emph{Proc. IEEE Symposium on Foundations of Computer Science (FOCS)}, pp.~285--294, 2005.
\url{https://doi.org/10.1109/SFCS.2005.29}.

\bibitem{Schmidt2020} K.-U.~Schmidt,
``Quadratic and symmetric bilinear forms over finite fields and their association schemes,''
\emph{Algebraic Combinatorics}, 3(1):161--189, 2020.
\url{https://doi.org/10.5802/alco.88}.

\bibitem{Sudan1997} M.~Sudan,
``Decoding of Reed--Solomon codes beyond the error-correction bound,''
\emph{Journal of Complexity}, 13(1):180--193, 1997.
\url{https://doi.org/10.1006/jcom.1997.0439}.
\end{thebibliography}
\end{document}